\documentclass[11pt]{article}

\usepackage{amsmath,amscd,amsfonts,eucal,latexsym,amssymb,mathrsfs}
\usepackage{amsxtra, amsthm, calc, graphics, color}
\theoremstyle{definition}
\newtheorem{Definition}{Definition}[section]
\theoremstyle{plain}
\newtheorem{Theorem}[Definition]{Theorem}
\newtheorem{Proposition}[Definition]{Proposition}
\newtheorem{Lemma}[Definition]{Lemma}
\newtheorem{Corollary}[Definition]{Corollary}
\theoremstyle{remark}
\newtheorem{Remark}[Definition]{Remark}

\numberwithin{equation}{section}
\newcounter{remcount}

\def\cA{{\cal A}}

\def\cC{{\cal C}}
\def\cD{{\cal D}}

\def\cF{{\cal F}}
\def\cG{{\cal G}}
\def\cH{{\cal H}}

\def\cK{{\cal K}}

\def\cP{{\cal P}}

\def\cS{{\cal S}}

\def\cU{{\cal U}}

\def\cZ{{\cal Z}}

\def\bC{{\mathbb C}}

\def\bN{{\mathbb N}}

\def\bR{{\mathbb R}}

\def\bZ{{\mathbb Z}}

\def\a{\alpha}
\def\b{\beta}
\def\g{\gamma}        
\def\d{\delta}        
\def\eps{\varepsilon}  

\def\k{\kappa}
\def\l{\lambda}       \def\L{\Lambda}
\def\m{\mu}
\def\n{\nu}

\def\r{\rho}
\def\s{\sigma}
    
\def\o{\omega}        \def\O{\Omega}

\def\fA{{\mathfrak A}}

\def\fW{{\mathfrak W}}

\def\supp{{\text{supp}\,}}
\def\ad{{\rm Ad}}
\newcommand{\rest}{\upharpoonright}
\DeclareMathOperator{\aut}{Aut}

\newcommand{\Id}{{\bf 1}}
\renewcommand{\Im}{{\mathrm{Im}\,}}
\renewcommand{\Re}{{\mathrm{Re}\,}}

\newcommand{\wlim}{\operatorname{w-lim}\displaylimits}
\newcommand{\bdx}{{\boldsymbol{x}}}
\newcommand{\bdy}{{\boldsymbol{y}}}
\newcommand{\bdxi}{{\boldsymbol{\xi}}}
\newcommand{\bdet}{{\boldsymbol{\eta}}}
\newcommand{\bdp}{{\boldsymbol{p}}}
\newcommand{\bdq}{{\boldsymbol{q}}}
\newcommand{\bds}{{\boldsymbol{s}}}
\newcommand{\uF}{\underline{F}}

\newcommand{\uA}{\underline{A}}
\newcommand{\uB}{\underline{B}}

\newcommand{\uW}{{\underline{W}}}

\newcommand{\uFF}{\underline{\mathfrak F}}
\newcommand{\uAA}{\underline{\mathfrak A}}
\newcommand{\uBB}{\underline{\mathfrak B}}
\newcommand{\ua}{\underline{\alpha}}

\newcommand{\uooi}{\underline{\o}_{0,\iota}}

\newcommand{\Aoi}{\mathcal{A}_{0,\iota}}

\newcommand{\uo}{\underline{\omega}}

\newcommand{\poi}{\pi_{0,\iota}}
\newcommand{\Hoi}{{\cal H}_{0,\iota}}

\newcommand{\Ooi}{\Omega_{0,\iota}}
\newcommand{\ooi}{\omega_{0,\iota}}
\newcommand{\uAAac}{\uAA^\bullet}

\newcommand{\poiac}{\poi^\bullet}

\newcommand{\rac}{\boldsymbol{\rho}^\bullet}
\newcommand{\pac}{\boldsymbol{\phi}^\bullet}
\newcommand{\Aoir}{\cA_{0,\iota,r}}
\newcommand{\br}{\boldsymbol{\rho}}

\newcommand{\bp}{\boldsymbol{\phi}}

\newcommand{\am}{\a^{(m)}}
\newcommand{\az}{\a^{(0)}}
\newcommand{\om}{\o^{(m)}}
\newcommand{\ooim}{\om_{0,\iota}}
\newcommand{\oz}{\o^{(0)}}
\newcommand{\cAm}{\cA^{(m)}}
\newcommand{\cAz}{\cA^{(0)}}
\newcommand{\cFm}{\cF^{(m)}}
\newcommand{\cFz}{\cF^{(0)}}
\newcommand{\uAAm}{\uAA^{(m)}}
\newcommand{\uAAmac}{\uAA^{(m)\bullet}}
\newcommand{\uam}{\ua^{(m)}}
\newcommand{\Aoim}{\Aoi^{(m)}}

\newcommand{\Aoirm}{\Aoir^{(m)}}
\newcommand{\Foirm}{\cF_{0,\iota,r}^{(m)}}
\newcommand{\Hoim}{\Hoi^{(m)}}
\newcommand{\Koim}{\cK_{0,\iota}^{(m)}}
\newcommand{\Woi}{W_{0,\iota}}

\newcommand{\aoim}{\a^{(m;0,\iota)}}
\newcommand{\Vqn}{{V_n^{(q)}}}
\newcommand{\omp}{\o_m(\bdp)}

\newcommand{\cCm}{\cC^{(m)}}
\newcommand{\uCCmac}{\underline{\cC}^{(m)\bullet}}
\newcommand{\cCoim}{\cC^{(m)}_{0,\iota}}
\newcommand{\cCz}{\cC^{(0)}}

\newcommand{\Oz}{\Omega^{(0)}}
\newcommand{\Om}{\Omega^{(m)}}
\newcommand{\uCCm}{\underline{\cC}^{(m)}}
\newcommand{\fWm}{\fW^{(m)}}
\newcommand{\fWz}{\fW^{(0)}}
\newcommand{\bDoIm}{\overline{\cD_0(I)}^{\|\cdot\|_m}}

\newcounter{propcount}

\newlength{\maxlabelwidth}

\newcommand{\inst}[1]{$^\textrm{#1}$ }

\begin{document}

\title{Asymptotic morphisms and superselection theory \\ in the scaling limit II:
analysis of some models}
\author{Roberto Conti\inst{1} \and Gerardo Morsella\inst{2}}
\date{
\parbox[t]{0.9\textwidth}{\footnotesize{%
\begin{itemize}
\item[1] Dipartimento di Scienze di Base e Applicate per l'Ingegneria, 
Sapienza Universit\`a di Roma, via A. Scarpa, 16 I-00161 Roma (Italy), e-mail: roberto.conti@sbai.uniroma1.it
\item[2] Dipartimento di Matematica, Universit\`a di Roma Tor Vergata, via della Ricerca Scientifica, 1 I-00133 Roma (Italy), e-mail: morsella@mat.uniroma2.it
\end{itemize}
}}
\\
\vspace{\baselineskip}
\today}

\maketitle

\begin{abstract}
We introduced in a previous paper a general notion of asymptotic morphism of a given local net of observables, which allows to describe the sectors of a corresponding scaling limit net. Here, as an application, we illustrate the general framework by analyzing
the Schwinger model, which features confined charges. In particular, we explicitly construct asymptotic morphisms for these sectors in restriction to the subnet generated by the derivatives of the field and momentum at time zero. As a consequence, the confined charges of the Schwinger model are in principle accessible to observation.  We also study the obstructions, that can be traced back to the infrared singular nature of the massless free field in $d=2$, to perform the same construction for the complete Schwinger model net.  Finally, we exhibit asymptotic morphisms for the net generated by the massive free charged scalar field in four dimensions, where no infrared problems appear in the scaling limit.
\end{abstract}

\section{Introduction}

Quantum Field Theory (QFT), which is the term that generically describes a number of different approaches to the theory of fundamental interactions in particle physics, is a spectacular enterprise where the physical requirements meet the mathematical tools in a mix that 
since the late twenties has been a powerful ground for developing impressive crossing relations between the two disciplines.
The algebraic approach to QFT (AQFT) stands for its neat conceptual clarity and mathematical rigour, at least for what concerns the general structural analysis.
However, over the last few years it has been proven to be useful also  in providing a framework for constructing specific models.
In the original formulation,
the main object of  study of AQFT is a so-called {\it net of local observables}, namely a correspondence 
$$O \to \cA(O)$$
between certain regions of Minkowski spacetime and operator algebras \cite{Ta} acting on a fixed (vacuum) Hilbert space $\cH$, satisfying a bunch of physically meaningful axioms.
This point of view is exposed in full detail in R. Haag's book \cite{Ha} (see also \cite{R}).
The representation theory of this net accounts for the physical superselection sectors, i.e., the rules specifying which physical processes may take place and which 
are ruled out, e.g., by conservation laws. Thanks to a breakthrough result by Doplicher and Roberts~\cite{DR}, when carefully analyzed, this superselection structure uniquely identifies a global gauge group $G$ and a {field net} $O \to \cF(O)$ on a Hilbert space $\tilde\cH$, acted upon by $G$, such that $\cF(O)^G \simeq \cA(O)$ and, moreover, all the relevant representations of $\cA$ appear in $\tilde\cH$.

Another important ingredient for us is the concept of {\it scaling limit net} due to D. Buchholz and R. Verch \cite{BV}, which is the counterpart in the algebraic setting of the renormalization group analysis (as discussed also in~\cite{BDM}) and in particular allows one to attach an intrinsic meaning to degrees of freedom (like quarks and gluons) that are unobservable at the ambient scale, i.e. they are {\it confined}. 
In short, this is a net
$$O \to \Aoi(O)$$
(actually a family of such nets, selected through different limiting processes)
which is supposed to capture the short distance (or, equivalently, high energy) behaviour of the theory described by the net $O \to \cA(O)$ and, in a sense, is a sort of tangent space
(\`a la Gromov) of the original net.
The concept of scaling limit applies to $(\cF,G)$ as well, giving rise to a number of natural concepts like that of {\it preserved charge} \cite{DMV,DMV2}.
Moreover, the net $O \to \Aoi(O)$ displays its own superselection structure, and the relationship between the superselection structures of $\cA$ and $\Aoi$ provides the above mentioned intrinsic definition of confined sectors of $\cA$.

Since the sectors of $\cA$ may be described by suitable endomorphisms,
following an earlier suggestion by S. Doplicher we studied the possibility to describe the superselection sectors of $\Aoi$ in terms of some sort of {\it asymptotic endomorphisms}
of $\cA$. A general treatment of this topic appears in \cite{CM2}. The emerging mathematical concept resembles the so-called asymptotic morphisms of Connes-Higson in $E$-theory, a variant of Kasparov $KK$-theory, which is a cornerstone in the formulation of Noncommutative Geometry in the sense of A. Connes \cite{Co}. 
Similar ideas where quantum features (here, superselection rules) are described in terms of noncommutative geometry, have appeared ever since.
%

On the physical side, the interpretation of asymptotic morphisms can be understood by observing that, composing one of them with the vacuum, one obtains a family, indexed by the spatio-temporal scale $\l$, of neutral states of the original theory which for $\l \to 0$ approximate, in a suitable sense, a charged state of the scaling limit theory. This is of course reminiscent of the procedure of ``shifting a compensating charge behind the Moon'' by which one obtains charged states as limits of neutral ones at a fixed scale.

Such a picture is particularly interesting in connection with the theoretical problem, mentioned above, of the formulation of a physically meaningful notion of confined charge. As first pointed out in~\cite{Buc1}, the conventionally accepted approach to confinement relies on the comparison between the fundamental degrees of freedom used to defined the theory (e.g., gauge and 	Fermi fields in the Langrangian) and its scattering states. As such, it boils down essentially to attaching a physical interpretation to unobservable objects and therefore it can not have an intrinsic meaning. This is confirmed also by the fact that there are well known examples of theories whose observables can be obtained starting from very different sets of basic fields. On the contrary, as already recalled, a notion of confined charge which is entirely based on observables  can instead be obtained in the algebraic framework of QFT by combining the scaling limit construction with the superselection sectors analysis. In this setting, asymptotic morphisms 
can be used, at least in principle, to operationally decide if a given theory features confined charges. Indeed, the above mentioned states induced by asymptotic morphisms converge to eigenstates of the charge operator of the scaling limit theory. Therefore, it should be sufficient to test their values and dispersions at small scales on suitable observables converging in the scaling limit to the conserved current generating the appropriate charge.

In view of the above considerations, it seems desirable, both from the mathematical and the physical standpoint, to understand in some detail 
how the abstract 
concepts developed in \cite{CM2} fit with the analysis of some concrete models and to present explicit examples of asymptotic morphisms associated to confined sectors.

A popular toy model exhibiting the main features which are expected to characterize the confinement picture is the so-called Schwinger model, i.e., $d=2$ quantum electrodynamics with massless fermions. As it is well known, this model is exactly solvable, and its net of observables is isomorphic to the one generated by the free massive scalar field~\cite{LS}, i.e., the distributional solution of the Klein-Gordon equation $(\Box + m^2) \varphi = 0$. In view of the fact that the Coulomb energy of two opposite electric charges grows linearly with their mutual distance, the absence of charged sectors of the observable algebra  has been interpreted as a manifestation of confinement (see, e.g., \cite{BJ}).


For this specific model the scaling limit analysis has been carried through in~\cite{Buc1, BV2}. The final outcome is that the scaling limit net of the free massive scalar field in $d=2$ spacetime dimensions contains the corresponding massless net. It is worth mentioning here that the usual technical complications due to the infrared singularity of the massless free field in $d=2$ are overcome in this case by describing it in terms of Weyl operators. Then, as discussed already in~\cite{StWi,Cio}, the latter model exhibits non-trivial sectors (localizable in wedges), which are therefore confined sectors of the Schwinger model net in the language of~\cite{DMV2}. 
(We also point out references~\cite{HL,DeMe} for a discussion of many other features of the $d=2$ massless free field.)

Hence, in the present work we discuss the Schwinger  model from the point of view of our previous paper~\cite{CM2}, and in particular we provide an explicit construction of its asymptotic morphisms.
%
%

In conclusion, despite some technicalities (infrared problems, choice of subnets, several kinds of limits), probably the most important message that can be read off from this work is that 
it provides evidence to the fact that
in principle
confined charges are indeed accessible to observation, at least in some idealized sense. 

\medskip
We summarize the content of the paper.
In Section \ref{sect:Free}  we mainly fix the notation and the general background. Namely, we recall the construction of the scaling limit and of the asymptotic morphisms associated to its sectors. We also introduce the Weyl algebra and use it to define the local nets associated to free fields of different masses in $d=2$. The section ends with a description of a family of sectors of the scaling limit $\Aoim$ of the massive free field.
Section \ref{sect:asymptotic} contains the main results of this work.
It is devoted to the explicit construction of asymptotic morphisms for the net $\cCm$ generated by the derivatives of the time zero massive fields. 
Such asymptotic morphisms correspond, in the sense of \cite{CM2}, to the sectors of the scaling limit of $\cCm$ that are obtained by restricting those of $\Aoim$.
As a preliminary step towards this construction, we first provide a description of the scaling limit of $\cCm$ by showing that it basically coincides with the corresponding massless version $\cCz$.
In Section \ref{sec:Schwinger} we examine the possibility of exhibiting asymptotic morphisms directly for the sectors of $\Aoim$ by taking a similar route as in the previous section. This program can be partially carried on, however the properties of the emerging objects are weaker than those axiomatized in \cite{CM2}.
In Section \ref{sec:highdim} we consider the free charged massive scalar field in $d=4$ and show that in this case the construction of asymptotic morphisms for the sectors of the scaling limit net can be achieved
without any trouble, thus showing that the difficulties of Section \ref{sec:Schwinger} are due to the singular behaviour of the two-dimensional massless free field.
Finally, in Appendix \ref{app:quasiequiv}, we prove for the $d=2$ case the local quasiequivalence of the massive and massless vacuum states of the Weyl algebra in restriction to the
subalgebra generated by the derivatives of time zero fields, a statement which is needed in Section \ref{sect:asymptotic} and interesting on its own (cf. \cite{EF}).
A recent result also implying this fact has been independently obtained in \cite{BFR}.

\section{Free scalar fields and their scaling limits}\label{sect:Free}
For the convenience of the reader, we quickly recapitulate in this section the main results of~\cite{BV, BV2} about the abstract scaling limit construction and its application to the concrete model of the free scalar field, and of \cite{CM2} about the asymptotic morphisms associated to the sectors of the scaling limit theory.

\medskip

Let $O \mapsto \cA(O)$ a local net of von Neumann algebras indexed by open double cones $O \subset \bR^d$ and acting on a vacuum Hilbert space $\cH$, cf.~\cite{Ha, Ar2}. We assume that $\cA$ is covariant with respect to a unitary, strongly continuous representation $U : \cG \to \cU(\cH)$ satisfying the spectrum condition, where $\cG$ is a subgroup of the connected component $\cP_+^\uparrow$ of the Poincar\'e group containing the translations, and that there is a unique (up to a phase) translation invariant unit vector $\O$ (the vacuum). 
As usual, we will write $\a_{(\L,x)}=\ad(U(\L,x))$, $(\L,x) \in \cG$.
We indicate with $\cA_\text{loc}$ the union of the local algebras $\cA(O)$, and, by a slight abuse of notation, we also use $\cA$ for the quasi-local C*-algebra defined by the net, i.e.\ the norm closure of $\cA_\text{loc}$. Moreover, for more general possibly unbounded open regions $S \subset \bR^d$, $\cA(S)$ will denote the C*-algebra generated by all the $\cA(O)$ with $O \subset S$. 

\medskip

The local scaling algebra $\uAA(O)$ is then defined as the C*-algebra of all the bounded functions $\l \in (0, +\infty) \mapsto \uA_\l$ such that $\uA_\l \in \cA(\l O)$ for all $\l > 0$, and
\[
\| \ua_{(\L,x)}(\uA) - \uA \| := \sup_{\l > 0} \| \a_{(\L,\l x)}(\uA_\l) - \uA_\l \| \to 0 \qquad\text{as }(\L,x) \to (\Id, 0) \text{ in $\cG$}.
\]
Given a bounded function $\l \in (0, +\infty) \mapsto A_\l$ such that $A_\l \in \cA(\l O)$ for all $\l > 0$ and $h \in C_c(\bR^2)$ (continuous functions with compact support), it is convenient to set 
$$(\ua_h A)_\l= \int_{\bR^2} dx \, h(x) \a_{\l x}(A) \ ,  \quad \l > 0 $$
(with the integral defined in the strong sense),
which defines an element in $\uAA(O+\supp h)$.

Given then a locally normal state $\o$ on $\cA$ (e.g., $\o = \langle \O, (\cdot)\O\rangle$), one can consider the states $(\uo_\l)_{\l > 0}$ on $\uAA$ defined by $\uo_\l(\uA) := \o(\uA_\l)$, and the set of their weak* limit points as $\l \to 0$, which is actually independent of the original state $\o$. For any such limit state $\uooi$, the corresponding scaling limit net is then defined by
\[
\Aoi(O) := \poi(\uAA(O))'',
\]
with $(\poi, \Hoi, \Ooi)$ the GNS representation determined by $\uooi$. This new net satisfies the same structural properties of the underlying net $\cA$, possibly apart from uniqueness of the vacuum if $d = 2$.

In order to formulate the notion of asymptotic morphisms of $\cA$, we also need to introduce the net of C*-algebras $O \mapsto \uAAac(O)$, defined as the C*-algebras of bounded functions $\l \in (0,+\infty) \mapsto A_\l \in \cA(\l O)$ such that for all $\hat O \Supset O$, i.e., $\hat O \supset \bar O$, and for all $\eps > 0$, there exist elements $\uA, \uA' \in \uAA(\hat O)$ for which
\[
\limsup_\k \| (A_{\l_\k} - \uA_{\l_\k})\O\| + \|(A^*_{\l_\k}-\uA'_{\l_\k})\O \| < \eps,
\]
where $(\l_\k)_{\k \in K}$ ($K$ some index set) is a net, fixed once and for all, such that $\uooi = \lim_\k \uo_{\l_\k}$. It is then clear that $\uAA(O) \subset \uAAac(O)$, and one finds that $\poi$ extends to a morphism $\poiac : \uAAac \to \Aoi$. Moreover, under the mild assumption that $\cA$ has a convergent scaling limit~\cite[Def.\ 4.4]{BDM2}, there also holds $\Aoi(O) \subset \poiac(\uAAac(O))$.
 
We can now define a (tame) asymptotic morphism of $\cA$ (relative to the scaling limit state $\uooi = \lim_\k \uo_{\l_\k}$) as a family of maps $\r_\l : \cA \to \cA$, $\l > 0$, such that, for all $A, B \in \cA$, $\a \in \bC$, there holds
\begin{align*}
&\lim_\k \| [\r_{\l_\k}(A+\a B) -\r_{\l_\k}(A)-\a \r_{\l_\k}(B)]\O \| = 0,\\
&\lim_\k \| [\r_{\l_\k}(AB) - \r_{\l_\k}(A)\r_{\l_\k}(B)]\O \| = 0,\\
&\lim_\k \| [\r_{\l_\k}(A)^*-\r_{\l_\k}(A^*)]\O \| = 0,
\end{align*}
and moreover such that for all $A \in \cA$ the function $\l \mapsto \rac(A)_\l := \r_\l(A)$ belongs to $\uAAac$, the resulting map $\rac : \cA \to \uAAac$ is norm continuous, and
\[
\poiac\Big(\rac\Big(\bigcup_O \cA(O)\Big)\Big) \subset \bigcup_O \Aoi(O).
\]
In particular, an asymptotic isomorphism is an asymptotic morphism $(\phi_\l)$ such that the map $\pac : \cA \to \uAAac$ is injective and there exists a continuous section $\bar s : \Aoi \to \uAAac$ of $\poiac$ for which
\[
\pac\Big(\bigcup_O \cA(O)\Big) = \bar s\Big(\bigcup_O \Aoi(O)\Big).
\]
With these definitions, the main result of~\cite{CM2} states that if $\cA$ has convergent scaling limit and its quasi-local C*-algebra is isomorphic to $\Aoi$, there is a 1-1 correspondence between unitary equivalence classes of morphisms $\r_0 : \Aoi \to \Aoi$ such that $\r_0(\bigcup_O \Aoi(O))\subset \bigcup_O \Aoi(O))$ and naturally defined equivalence classes of pairs of an asymptotic morphism $(\r_\l)$ and an asymptotic isomorphism $(\phi_\l)$, such correspondence being defined by the formula
\[
\r_0 = \poiac \rac (\pac)^{-1}\bar s.
\]
We notice explicitly that the above definitions and results make sense in any number of spacetime dimensions $d$.

We now turn to the description, following~\cite{BV2}, of the scaling limit of the free scalar field, focusing on the $d=2$ case. We denote by $\fW$ the Weyl algebra, the C*-algebra generated by the unitary operators $W(f)$, $ f \in \cD(\bR)$ (complex valued functions), satisfying
\begin{gather*}
W(f)W(g) = e^{-\frac{i}{2}\s(f,g)}W(f+g),  \\
\s(f,g) = \Im \int_\bR d\bdx\,\overline{f(\bdx)}g(\bdx).
\end{gather*}
For each mass $m \geq 0$, there is an automorphic action of the Poincar\'e group $\cP = O(1,1) \ltimes \bR^2$ on $\fW$, denoted by $(\Lambda,x) \mapsto \am_{(\L,x)}$ and induced by an action $\tau^{(m)}_{(\L,x)}$ on $\cD(\bR)$, i.e., $\am_{(\L,x)}(W(f))=W(\tau^{(m)}_{(\L,x)}f)$. For reference's sake, we give the explicit expression of time translations:
\begin{equation}\label{eq:timeev}\begin{split}
(\tau^{(m)}_t f)\hat{}(\bdp) &= \left[\cos(t\omp) + i\omp^{-1}\sin(t\omp)\right] (\Re f)\hat{}(\bdp) \\
&\quad+ i\left[\cos(t\omp)+i\omp \sin(t\omp)\right](\Im f)\hat{}(\bdp), \qquad t \in \bR,
\end{split}\end{equation}
where $\omp=\sqrt{\bdp^2 + m^2}$.
There is also an automorphic action $\l \in \bR_+ \mapsto \s_\l$ of dilations on $\fW$, induced by an action $\l \mapsto \d_\l$ on $\cD(\bR)$, see~\cite[Eq.\ (2.7)]{BV2}. We also consider the vacuum states $\om$, $m \geq 0$, on $\fW$. For $m > 0$ they are defined by $\om(W(f)) = e^{-\frac 1 2 \|f\|_m^2}$, where
$$\|f\|_m^2 = \frac{1}{2}\int_\bR d\bdp\left|\o_m(\bdp)^{-1/2} (\Re f)\hat{}(\bdp) + i\o_m(\bdp)^{1/2} (\Im f)\hat{}(\bdp)\right|^2, \qquad m \geq 0,$$
while, for $m=0$,
\begin{equation*}
\oz(W(f)) = \begin{cases}e^{-\frac{1}{2}\| f\|_0^2} &\text{if }(\Re f)\hat{}(0) = 0,\\
0 &\text{otherwise.}\end{cases}
\end{equation*} 
It turns out that $\|\cdot\|_m$ for $m>0$ (resp.\ $m=0$) is indeed a norm on $\cD(\bR)$ (resp.\ $\{f \in \cD(\bR) \ : \ \int_\bR \Re f = 0\}$) considered as a real vector space. We denote by $\pi^{(m)}$  the GNS representation  induced by $\o^{(m)}$, $m \geq 0$, acting on the Hilbert space $\cH^{(m)}$ with cyclic vector $\Om$, and  by $O \mapsto \cAm(O)$ the corresponding net of von Neumann algebras 
$$
\cAm(O) := \{\pi^{(m)} (\am_{(\L,x)}(W(f)))\,: \supp f \subset I\}'',
$$
where $I \subset \bR$ is an open interval such that $O = \L O_I+x$, with $O_I$ the double cone with base $I$ in the time-zero line.
(Note that here we depart from the notation of~\cite{BV2}; also, in order to simplify the notation, we will drop the indication of the representation $\pi^{(m)}$ when this does not cause confusion.)
It is clear that the net $\cAm$ 
satisfies
\begin{equation}\label{eq:innercont}
\cAm(O_I) = \bigvee_{J \Subset I} \cAm(O_J) \ ,
\end{equation}
since every $f$ supported in an open interval $I$ is actually supported in a $J \Subset I$.
 If $m > 0$, the net $\cAm$ satisfies the split property~\cite{Dr} and its local algebras are type III$_1$ factors \cite{GlJa}. 
 In particular, $\cAm$ is a simple C*-algebra. Note also that $\cH^{(0)}$ is non-separable, see \cite[Section 4]{AMS93}.\footnote{A representation of $\fW$ in which $\tau^{(0)}$ is unitarily implemented on a separable Hilbert space is constructed in~\cite{DeMe}.}

According to the general construction discussed above, we associate to $\cAm$, $m>0$, the scaling algebra $\uAAm$ and the scaling limit nets $O \mapsto \Aoim(O) = \poi(\uAAm(O))''$, on the scaling limit Hilbert spaces $\Hoim$, along with corresponding automorphic actions $\aoim$ of the Poincar\'e group. We will also make use of the net $\uAAmac$ of the elements asymptotically contained in $\uAAm$, and of the corresponding extension $\poiac$ of the scaling limit representation. When there is no ambiguity, we will just write $\uAAac$ instead of $\uAAmac$. It is known that, for each $\iota$, the quasi-local algebra has a non-trivial center, and the local algebras $\Aoim(O)$ are not factors, as they contain non-trivial elements of the center~\cite[proof. of Thm. 4.1]{BV2}. As a consequence, since $\cAm$ is irreducible (because the Fock vacuum is the only translation invariant vector), the quasi-local C*-algebras $\Aoim$ and $\cAm$ cannot be isomorphic. This is in sharp contrast with the situation considered in~\cite{CM2}.

One can also introduce rescaled Weyl operators
$$\uW(f)_\l := W(\d_\l f), \qquad \l>0,\, f \in \cD(\bR).$$
From~\cite[Eq. (4.13)]{BV2} one sees that
\begin{equation}\label{eq:weylreg}
\limsup_{\l \to 0} \| [\uW(f)_\l - (\uam_h\uW(f))_\l]\O\|
\end{equation}
can be made arbitrarily small choosing $h$ sufficiently close to a $\d$-function, which entails $\uW(f) \in \uAAac(O)$ for every double cone $O$ based on the time zero line and whose base contains $\supp f$. Using~\cite[Lemma 4.2]{BV2} and~\cite[Lemma 4.5]{CM2}, one also concludes that
\begin{equation}\label{eq:Woi}
\Woi(f) := \poiac(\uW(f) ), \qquad f \in \cD(\bR),
\end{equation}
satisfy the Weyl relations
$$\Woi(f) \Woi(g) = e^{-i\sigma(f,g)/2}\Woi(f+g), \qquad f,g\in\cD(\bR),$$
and $\Woi(f) \in \Aoim(O_I)$ if $\supp f \subset I$.
The arguments in~\cite{Buc1} suggest that the net $\cAz$ obtained in the GNS representation of the massless vacuum $\oz$ is isomorphic to the subnet of $\Aoim$ generated by the operators $\Woi(f)$.\footnote{The results in~\cite{Buc1} also suggest that the scaling limit of the C*-subalgebra of the scaling algebra generated by (smoothed-out) functions $\l \mapsto W(\d_\l f)$ and $\l \mapsto W(|\log \l|^{1/2}\d_\l f)$ is isomorphic to $\cAz \otimes \cZ$, with $\cZ$ (a subalgebra of) the center of $\Aoim$.} In Sec.~\ref{sec:Schwinger} we will make this identification more explicit. 

The net $\Aoim$ has non-trivial automorphisms $\r_{q,\iota}$, $q \in \bR$, defined as follows. Let $u_n^q = u_n \in \cD_\bR(\bR)$ be such that, for some $a>0$,
$$u_n(\bdx) = \begin{cases}0 \qquad &\text{if }\bdx \leq -a,\\
\text{independent of }n &\text{if }-a < \bdx < a,\\
q &\text{if }a \leq \bdx \leq na,\end{cases}
$$
and consider $\Vqn := \Woi(iu_n)$. It follows that, for $m \geq n$, $\Vqn V_m^{(q)*} =\Woi(i(u_n-u_m))$ is localized in a double cone 
whose closure is contained 
in the right spacelike complement of $\bdx = na$. Therefore, if $A \in \Aoim(O)$, $\Vqn^*A\Vqn$ is independent of $n$ for $n$ sufficiently large, and there exists the (norm) limit
\begin{equation}\label{eq:rqi}
\r_{q,\iota}(A) = \lim_{n \to +\infty}\Vqn^*A\Vqn.
\end{equation}
It is then clear that $\r_{q,\iota}$ extends, by norm continuity, to an endomorphism of the quasi-local algebra $\Aoim$, which is easily seen to be 
localized in $W_+-a$ 
and
invertible. It is shown in~\cite{BV2} that such automorphisms induce non-trivial translation covariant sectors of $\Aoim$ (which are independent of the chosen $a > 0$ and of $u_n$ on $(-a,a) \cup (na,+\infty)$). More precisely, it is shown in \cite[Thm.\ 4.1]{BV2} that $\ooi\circ\r_{q,\iota} \rest \Aoim(W_a^\pm)  = \ooi \rest \Aoim(W_a^\pm)$, where $W_a^\pm$ is the right/left spacelike complement of the interval $[-a,a]$ of the time-zero axis. From this it follows, by cyclicity of $\Ooi$ for wedge algebras and GNS unicity, that $\r_{q,\iota}$ is a BF-sector of $\Aoim$. Notice also that, due to the above mentioned localization properties of $\Vqn$, $\r_{q,\iota}$ is a properly supported morphism of $\Aoim$ in the sense of~\cite[Sec.\ 5]{CM2}.

\section{Asymptotic morphisms for the derivatives of time zero fields} \label{sect:asymptotic}
Since the sectors $\r_{q,\iota}$ can be interpreted as describing confined charges of the underlying theory $\cAm$ \cite{Buc1,DMV}, it seems interesting to 
exhibit explicit examples of the associated asymptotic morphisms of $\cAm$. 
Actually, due to the bad infrared behaviour of the massless scalar field in $d=2$ which is responsible for the appearance of the nontrivial center of $\Aoim$, 
these sectors do not fall into the framework of \cite{CM2}. Therefore, 
it will only be possible to construct the associated asymptotic morphisms at the price of passing from $\cAm$ to a suitable subnet $\cCm$  generated by the derivatives of the time zero fields (see Eq. \eqref{eq:Cm} below for a precise definition), which has the property that its scaling limit inherits the same sectors from $\Aoim$.

\medskip

To begin with, 
we compute for reference the action of $\r_{q,\iota}$  on the Weyl operators defined above:
\begin{equation*}\begin{split}
\r_{q,\iota}(\Woi(f)) &= \lim_{n \to +\infty} \Woi(iu_n)^*\Woi(f)\Woi(iu_n) \\
&= \lim_{n \to +\infty} e^{i\s(iu_n,f)}\Woi(f)\\
&=\lim_{n \to +\infty}e^{i \Im \int_\bR d\bdx\,(-i)u_n(\bdx)f(\bdx)}\Woi(f)\\
&=\lim_{n \to +\infty}e^{-i\int_\bR d\bdx\,u_n(\bdx)\Re f(\bdx)}\Woi(f)\\
&= e^{-i \int_\bR d\bdx\,u_\infty(\bdx)\Re f(\bdx)}\Woi(f),
\end{split}\end{equation*}
where $u_\infty = \lim_{n\to +\infty} u_n$ is such that $u_\infty(\bdx) = 0$ if $\bdx < -a$ and $u_\infty(\bdx) = q$ if $\bdx > a$, cf.~\cite[Sec. 4]{Cio}.

We also observe that similar formulas hold for the rescaled Weyl operators at each fixed $\l > 0$.
Namely, by the same argument as above, we can define, for each $\l > 0$, a morphism $\r(\l)$ of $\cAm$ by
\begin{equation*}
\r(\l)(A) = \lim_{n \to +\infty}\uW(iu_n)^*_\l A \uW(iu_n)_\l, \qquad A \in \cAm,
\end{equation*}
and there holds
\begin{equation*}\begin{split}
\r(\l)(W(f)) &= \lim_{n \to +\infty} \uW(iu_n)^*_\l W(f) \uW(iu_n)_\l \\
&= \lim_{n \to +\infty}W(\d_\l iu_n)^*W(f)W(\d_\l iu_n) \\ 
&= \lim_{n \to +\infty}e^{-i\int_\bR d\bdx\,\d_\l u_n(\bdx)\Re f(\bdx)} W(f)\\
&= \lim_{n \to +\infty}e^{-i\int_\bR d\bdx\,u_n(\l^{-1}\bdx)\Re f(\bdx)} W(f)\\
&= e^{-i \int_\bR d\bdx\,u_\infty(\l^{-1}\bdx)\Re f(\bdx)}W(f).
\end{split}\end{equation*}
We can also perform the limit $\l \to 0$ of the last line, obtaining
$$
\lim_{\l \to 0}\r(\l)(W(f)) =
\lim_{\l \to 0}\lim_{n \to +\infty} \uW(iu_n)^*_\l W(f) \uW(iu_n)_\l =e^{-i q\int_0^{+\infty} d\bdx\,\Re f(\bdx)}W(f).
$$
Moreover, we can extend the automorphisms $\s_\l : \fW \to \fW$ to automorphisms $\phi_\l : \cAm \to \cAm$ such that 
\begin{equation}\label{eq:phil}
\phi_\l(\pi^{(m)}(W(f))) = \pi^{(m)}\s_\l(W(f)) = \pi^{(m)}(W(\d_\l f)) = \pi^{(m)}(\uW(f)_\l), \qquad \l > 0, f \in \cD(\bR).
\end{equation}
This is a direct consequence of the local normality of the states $\o^{(m)}$ and $\o^{(m)}\circ \s_\l = \o^{(\l m)}$~\cite{EF}.
 Therefore, if we define the morphisms $\r_\l := \r(\l)\phi_\l : \cAm \to \cAm$, we obtain, from the above formulas,
\begin{equation}\label{eq:rhol}
\r_\l(W(f)) =  e^{-i \int_\bR d\bdx\,u_\infty(\l^{-1}\bdx)\Re (\d_\l f)(\bdx)}W(\d_\l f) =  e^{-i \int_\bR d\bdx\,u_\infty(\bdx)\Re f(\bdx)}\uW(f)_\l.
\end{equation}
Identifying $\fW$ with the C$^*$-subalgebras of $\cAm$ and $\Aoim$ generated by the operators $W(f)$, $\Woi(f)$ ($f \in \cD(\bR)$) respectively, we can also consider an isomorphism $\bp : \fW \to \fW$ such that
$$\bp(W(f)) = \Woi(f) \ . $$
Therefore, defining $\br = \r_{q,\iota} \bp: \fW \to \Aoim$, we get
$$\poiac\big(\l \mapsto \r_\l(W(f))\big) = \br(W(f)), \qquad  f \in \cD(\bR) $$
from which, using the fact that, as $\r_\l$ is a morphism of C*-algebras,  $\| \r_\l\| \leq 1$ uniformly for $\l > 0$, and the norm continuity of $\poiac$ and $\br$, we conclude that (cf.~\cite[Thm. 5.2]{CM2})
$$\poiac\big(\l \mapsto \r_\l(W)\big) = \br(W), \qquad W \in \fW.$$
In passing, we also note that $(\r_\l\rest \fW)$ satisfies properties analogous to properties (i)-(iii) of~\cite[Def. 5.1]{CM2} (in particular, for property (iii), replacing $\bigcup_O \cAm(O)$ with $\bigcup_I \fW_I$, where $\fW_I$ is the Weyl algebra over $\cD(I)$, $I \subset \bR$ interval). 

Moreover, it is worth pointing out that, at variance with the case of preserved sectors considered in~\cite[Sec. 7]{CM2}, the morphisms $\r(\l)$ do not induce non-trivial sectors of $\cAm$, because there are no such sectors~\cite[Thm.\ 3.1 and Sec.\ 7]{Mu}. The latter fact in particular implies that the sectors induced by $\r_{q,\iota}$ are interpreted as confined sectors of $\cAm$~\cite{Buc1, DMV, DMV2}.

\medskip


The family of maps $(\r_\l)$ defined above enjoys some of the properties of asymptotic morphisms of $\cAm$, but it can be shown that $\l \mapsto \r_\l(A)$ does not belong to $\uAAmac$ for all $A \in \cAm$. We defer a discussion of this and related aspects to Section \ref{sec:Schwinger}, where it will appear that the real obstruction lies in the zero mode of the field momentum. Therefore, for the time being, 
in order to achieve our goal of constructing \emph{bona fide} asymptotic morphisms,
we restrict our attention to the von Neumann algebras
\begin{equation}\label{eq:Cm}
\cCm(O_I +x) := \{\pi^{(m)}(\am_{x}(W(f)))\,: \supp f \subset I, \, {\textstyle \int_\bR f = 0}\}'', \qquad m > 0.
\end{equation}
This way, since the condition of having null integral is preserved by the translation $\tau^{(m)}_x$, $m > 0$, we obtain a
translation covariant subnet $\cCm$ of the restriction of $\cAm$ to the upright double cones (i.e., those of the form $O_I + x$).
However, the property $\int_\bR f = 0$ is not stable under Lorentz transformations \cite[Section 7.2]{BDM2}, and therefore $\cCm$ cannot be extended to a Poincar\'e covariant isotonous subnet of $\cAm$. 
This obstruction disappears for $m=0$, so that
\begin{equation}\label{eq:Cz}
\cCz(\L O_I +x) := \{\pi^{(0)}(\az_{(\L,x)}(W(f)))\,: \supp f \subset I, \, {\textstyle \int_\bR f = 0}\}'' 
\end{equation}
defines a Poincar\'e covariant subnet of $\cAz$ (indexed by all double cones). We denote by $\cK^{(0)} = \overline{\cCz \Oz}$ the corresponding cyclic subspace of $\cH^{(0)}$.

We note that the dual net local algebras $\cC^{(m)d}(O) := \cCm(O')'$ can be defined for all double cones $O$ (not just the upright ones). With this in mind, we summarize in the next statement the main relations between $\cCm$ and $\cAm$.

\begin{Proposition}\label{prop:Cm}
Let $m > 0$. Then the following properties hold.
\renewcommand{\theenumi}{(\roman{enumi})}
\renewcommand{\labelenumi}{\theenumi}
\begin{enumerate}
\item If $W_+$ denotes the right wedge, $\cCm(W_+)'' = \cAm(W_+)''$.
\item $\cC^{(m)d} = \cAm$.
\item There is an action $\gamma : \bR^2 \to \aut(\cAm)$ by net automorphisms, such that
\[
\g_{(\m,\n)}(\pi^{(m)}(W(f))) = e^{i(\m\,\Re \int_\bR f + \nu \,\Im \int_\bR f)}\pi^{(m)}(W(f)),
\]
and $\cCm(O) \subset \cAm(O)^{\bR^2}$ for all upright double cones $O$ .
\end{enumerate}
\end{Proposition}

\begin{proof}
(i) Let $f \in \cD(\bR)$ have support contained in $\bR_+$ and set $\a := \int_\bR f$. Moreover, let $\chi \in \cD(\bR)$ be a real function with support in $(0,1)$ and $\int_\bR \chi = 1$, and consider the function $f_\eps(\bdx) := f(\bdx) - \alpha \eps \chi(\eps \bdx)$, $\eps > 0$. Then clearly $\supp f_\eps \subset \bR_+$, $\int_\bR f_\eps = 0$ and
\[\begin{split}
\|f-f_\eps \|^2_m &=  \int_\bR d\bdp\left|\frac {\Re \a} {\o_m(\bdp)^{1/2}} + i \o_m(\bdp)^{1/2} \Im \a\right|^2 |\hat\chi(\bdp/\eps)|^2 \\
&= \int_\bR d\bdp\left[\frac {(\Re \a)^2} {\o_{m/\eps}(\bdp)} +  \o_{\eps m}(\eps^2 \bdp) (\Im \a)^2\right] |\hat\chi(\bdp)|^2 \to 0
\end{split}\]
as $\eps \to 0$, by a straightforward application of the dominated convergence theorem. This implies that $W(f_\eps) \in \cCm(W_+)$ converges strongly to $W(f) \in \cAm(W_+)$, and therefore the statement.

(ii) Given a double cone $O = (W_++a)\cap (-W_++b)$ one has
\[\begin{split}
\cC^{(m)d}(O) &= \cCm(O')' = \cCm(-W_++a)' \wedge \cCm(W_++b)'\\
&= \cAm(-W_++a)' \wedge \cAm(W_++b)' = \cA^{(m)d}(O) = \cAm(O),
\end{split}\]
where (i) and duality for the net $\cAm$ have been used. 

(iii) Let $O = O_I$, and let $f \in \cD(I)$. Given $g \in \cD(\bR)$ such that $g|_I = -i$, one has, by the Weyl relations,
\begin{equation}\label{eq:weylauto}
W(\m g) W(f) W(\m g)^* = e^{i \m \Re\int_\bR f}W(f),
\end{equation}
which shows that $\gamma_{(\m,0)}$ can be defined as an automorphism of the von Neumann algebra $\cAm(O_I)$. Then, it extends to an automorphism of the quasi-local algebra, as the latter is generated, as a C*-algebra, by local algebras of this form. By the same argument, using a function $h \in \cD(\bR)$ such that $h|_I = 1$, one gets $\gamma_{(0,\n)}$. Moreover, it is clear, again by the Weyl relations, that the group-like commutator
\[
W(\n g)^* W(\m h)^* W(\n g) W(\m h)
\]
is a phase factor, thus yielding the required automorphic action of $\bR^2$ (cf.~\cite[Sec. 5]{AMS93}). By~\eqref{eq:weylauto}, it is also clear that $\g_{(\m,\n)}(\cAm(O)) = \cAm(O)$ for all upright
double cones $O$ (not necessarily based at time zero). Finally, the $\g$-invariance of elements of $\cCm(O_I)$ is obvious, and the general case readily follows by the fact that every upright double cone is included in one based on the time zero line.
\end{proof}

Property (iii) above entails in particular that $\cCm$ is a proper subnet of $\cAm$ (restricted to the upright double cones). We also remark that the automorphisms $\g_{(\m,\n)}$ commute with spatial translations. 

We will also need to consider the Poincar\'e covariant subnet $\cC$ of $\cAz$ defined by
\begin{equation}\label{eq:C}
\cC(\Lambda O_I + x) := \Big\{ \pi^{(0)}(\az_{\Lambda,x}(W(f)))\,:\,\supp f \subset I, {\textstyle\int_\bR \Re f = 0}\Big\}''.
\end{equation}
Here the isotony is guaranteed by the fact that massless Poincar\'e transformations also preserve the property $\int_\bR \Re f = 0$. This net is Haag-dual by~\cite[Appendix 3]{HL} when considered as a net acting on the cyclic Hilbert space defined by $\Oz$. Moreover, it enjoys the split property and the local algebras $\cC(O)$ are hyperfinite type III$_1$ factors. This follows from the fact that, as shown for instance in~\cite{BLM}, the local algebras decompose into a tensor product of local algebras associated to the U(1) current on the light rays (for the latter see, e.g., \cite{BMT}). We will see in Sec.~\ref{sec:Schwinger} that $\cC$ can be naturally identified with a Poincar\'e covariant subnet of $\Aoim$. Also, by an argument similar to the one in Prop.~\ref{prop:Cm}, we obtain that $\cC^{(0)d} = \cC$. In particular, $\overline{\cC\Oz} = \cK^{(0)}$.

According to~\cite[Thm.\ 3.1]{CM2}, there exists then a C*-algebra isomorphism $\phi : \cAm \to \cC$ which identifies a countable increasing family of local von Neumann algebras. In view of~\cite{EF}, a natural question would then be if $\phi$ can be chosen in such a way that $\phi(\pi^{(m)}(W(f))) = \pi^{(0)}(W(f))$ whenever $\int_\bR \Re f = 0$. That this can not be the case follows at once from the fact, noted below in Remark \ref{notiso}, that the algebra~\eqref{eq:Ctildem} defined there is a proper subalgebra of $\cAm(O)$ for any $O$. Similarly, there is no isomorphism between $\cAm(O)$ and $\cA^{(0)}(O)$ mapping $\pi^{(m)}(W(f))$ into $\pi^{(0)}(W(f))$ for all $f \in \cD(\bR)$: indeed, the map ${\mathbb R} \ni \alpha \mapsto \pi^{(m)}(W(\alpha f))$ is $\sigma$-strongly continuous for every $f$,
while $\alpha \mapsto \pi^{(0)}(W(\alpha f))$ is not, and any von Neumann algebra isomorphism is automatically continuous for the $\sigma$-strong topologies.
However, at this stage we can not exclude the existence of an isomorphism between~\eqref{eq:Ctildem} and $\cC(O)$, defined by the same formula. 
On the positive side, it follows from results in \cite{BFR}
that for any interval $I \subset \bR$ there exists a von Neumann algebra isomorphism
$\phi :\cC^{(m)}(O_I) \to \cC^{(0)}(O_I)$ such that $\phi(\pi^{(m)}(W(f))) = \pi^{(0)}(W(f))$ if $\supp f \subset I$ and $\int_\bR f = 0$. 
For the sake of self-containment, a direct proof of this result is provided in Appendix~\ref{app:quasiequiv}.
This is an extension to $d=2$ of the classical result of~\cite{EF}, compatible with the infrared singularity of the massless free field. 
\smallskip

We will need a result about a certain phase space property for the massless scalar field in two dimensions that might be of independent interest.

\begin{Proposition}
The net $\cC$ on ${\mathbb R}^2$ in its vacuum representation on $\cK^{(0)}$ satisfies the Buchholz-Wichmann nuclearity condition, namely  the map
\[
\Theta : \cC(O) \to \cK^{(0)},\qquad A \mapsto e^{-\b H} A \Oz,
\]
is nuclear for all $\beta > 0$ and $O$.
\end{Proposition}

\begin{proof}
The local algebras $\cC(O)$ can be written as tensor products of two local algebras of the $U(1)$ current algebra on ${\mathbb R}$ relative to suitable intervals $I,J$. Now, it is well known that the $U(1)$ current algebra satisfies the trace-class condition and thus, thanks to \cite{BDL}, it satisfies the Buchholz-Wichmann nuclearity. Finally, it can be shown without too much trouble that the von Neumann tensor product of the Buchholz-Wichmann nuclear maps associated to the two intervals $I,J$ is again nuclear. To see this, observe that thanks to \cite[Lemma 2.2]{BDF} one can write $\Theta_I = \sum_n f_{n,I}(\cdot) \xi_{n,I}$ for some normal functionals $f_{n,I}$ on the local von Neumann algebra of the $U(1)$ current associated to the interval $I$ and some Hilbert space vectors $\xi_{n,I}$ such that $\sum_n \|f_{n,I}\| \, \|\xi_{n,I}\| < +\infty$. Now, by \cite[Sec. IV.5]{Ta}, $f_{n,I} \otimes f_{m,J}$ extends to a normal functional on $\cC(O)$ whose norm is equal to $\|f_{n,I}\| \, \| f_{m,J} \|$. It is now clear that
$\Theta = \sum_{n,m} (f_{n,I} \otimes f_{m,J})(\cdot) \xi_{n,I} \otimes \xi_{m,J}$ on the algebraic tensor product of the U(1)-von Neumann algebras of the intervals $I$ and $J$. Since both sides of this equality are continuous w.r.t.\ the $\sigma$-weak topology on the domain and the weak topology on the target, then they coincide as maps on $\cC(O)$ and furthermore $\sum_{n,m} \|f_{n,I} \otimes f_{m,J}\| \,  \|\xi_{n,I} \otimes \xi_{m,J}\| < +\infty$.
\end{proof}

Hereafter, as in~\cite{BV2}  it will be convenient to pass from the net $\cCm$ defined in~\eqref{eq:Cm}, to the one whose local algebras are
\[
\{\pi^{(0)}(\am_x(W( f)))\,: \supp f \subset I, \, {\textstyle \int_\bR f = 0}\}'',
\]
which is net-isomorphic to the former one thanks to Thm.\ \ref{thm:quasiequiv}. We will therefore assume that this has been done. Note that in particular, if $O_I$ is a double cone with basis on the time zero line, with this definition one has $\cCm(O_I) = \cCz(O_I)$.

\begin{Theorem}\label{thm:limitCm}
Let $\cCoim$ be a scaling limit net of $\cCm$, $m > 0$, acting on the Hilbert space $\Koim$. Then there exists a unitary operator $V : \Koim \to \cK^{(0)}$ satisfying the following properties:
\renewcommand{\theenumi}{(\roman{enumi})}
\renewcommand{\labelenumi}{\theenumi}
\begin{enumerate}
\item \label{it:defadV}for all $\uA \in \uCCm(O)$ there holds $\ad V(\poi(\uA)) = \lim_\k \phi_{\l_\k}^{-1}(\uA_{\l_\k})$ weakly;
\item \label{it:Vvacuum}$V\Ooi = \Oz$;
\item $\ad V\circ \aoim_{x} = \az_{x}$ for all translations $x \in \bR^d$; 
\item \label{it:adVCoim}for each pair of upright double cones $O \Subset \tilde O$,
\[
\cCz(O) \subset \ad V(\cCoim(\tilde O)) \subset \cCz(\tilde O);
\]
\item \label{it:adVpoi} given a double cone $O_I$ with basis on the time zero line and $A \in \cCz(O_I) = \cCm(O_I)$, there holds
\[
\ad V(\poi(\uam_h \pac(A) )) = \int_{\bR^2} dx\,h(x) \az_x(A)
\]
for all $h \in C_c(\bR^2)$;
\item 
for all $f \in \cD(\bR)$ such that $\int_\bR f = 0$ the Weyl operator $\Woi(f)$ leaves $\Koim$ invariant and
\[
\ad V (\Woi(f)) = \pi^{(0)} (W(f)) |_{\cK^{(0)}} \ . 
\]
\end{enumerate}
In particular, for the associated quasi-local algebras there holds $\ad V(\cCoim) = \cCz$.
\end{Theorem}

\begin{proof}
The proof closely follows the one of~\cite[Thm.~3.1]{BV2}, so we limit ourselves to point out the main differences. The key ingredient of that proof, namely the local normality of the vacuum states of the massive and massless free scalar field, is here replaced by the isomorphism of Thm.~\ref{thm:quasiequiv}. Moreover,
one has, for $f \in \cD(\bR)$ such that $\int_\bR f = 0$,
\begin{equation}\label{eq:normtranslation}\begin{split}
&\| \tau^{(\l m)}_x f - \tau_x^{(0)}f \|_0^2 = \|\tau_t^{(\l m)}f-\tau^{(0)}_tf\|^2_0\\
&= \frac12\int_\bR\frac{d\bdp}{|\bdp|}\bigg| \big[ \cos(t\o_{\l m}(\bdp))-\cos(t|\bdp|)\big] \widehat{\Re f}(\bdp)
-\big[  \o_{\l m}(\bdp)\sin(t\o_{\l m}(\bdp))-|\bdp|\sin(t|\bdp|)\big]\widehat{\Im f}(\bdp)\\
&\phantom{= \frac12\int_\bR\frac{d\bdp}{|\bdp|}}
+i|\bdp|\big[ \cos(t\o_{\l m}(\bdp))-\cos(t|\bdp|)\big] \widehat{\Im f}(\bdp)
+i|\bdp|\bigg[\frac{\sin(t\o_{\l m}(\bdp))}{\o_{\l m}(\bdp)}-\frac{\sin(t|\bdp|)}{|\bdp|}\bigg]\widehat{\Re f}(\bdp)\bigg|^2.
\end{split}\end{equation}
This integral 
is seen to converge to zero, as $\lambda \to 0$, by an application of the dominated convergence theorem, since the integrand can be bounded, for fixed $t \in \bR$ and all $\l \in [0,1]$, by the function
\begin{equation}\label{eq:estimate}
\frac{4}{|\bdp|}\left[(1+|t|\o_m(\bdp))\big|\widehat{\Re f}(\bdp)\big|+2\o_m(\bdp)\big|\widehat{\Im f}(\bdp)\big|\right]^2,
\end{equation}
which is integrable thanks to the fact that $\hat f$ is a Schwarz function such that $\hat f(0)=0$.  This shows that the key estimate in~\cite[Lemma 3.2(b)]{BV2} can be done also in our case.
 The fact that $\cC$ satisfies Buchholz-Wichmann nuclearity entails that this is true also for the subnet $\cCz$, and this allows us to repeat in our context the proof of~\cite[Lemma~3.3]{BV2}. The above arguments show the validity of (i)-(v). Finally, (vi) is obtained by observing that for all $\uA \in \uCCm(O_I)$ and $f$ as in the statement one has $\Woi(f) \poiac(\uA) \Ooi = \poiac(\uW(f)\uA)\Ooi \in \Koim$ and then computing directly the l.h.s. of the equality using (v).
\end{proof}

Taking into account the outer regularity of $\cCz$, consequence of the strong continuity of the action of the dilation group on $\cK^{(0)}$, it readily follows that $\ad V$ implements a net isomorphism between the outer regularized net of $\cCoim$ and $\cCz$, similar to the situation discussed in \cite{BV2} for the higher dimensional case. Of course, the quasi-local $C^*$-algebras of the net $\cCoim$ and of its outer regularized net coincide.
 
 \smallskip
 We also see from this concrete situation that the dual of the scaling limit net is not necessarily equal to the scaling limit net of the dual.

\begin{Remark}\label{notiso}
The above results also show that the scaling limit of $\cCm$ is a proper subnet of the net $\cC$ defined in~\eqref{eq:C}. In order to obtain the full net $\cC$ in the scaling limit, the first guess would be to associate to the double cone $O = \Lambda O_I +x$ the von Neumann algebra
\begin{equation}\label{eq:Ctildem}
 \{\am_{\L,x}(\pi^{(m)}(W(f)))\,: \supp f \subset I, \, {\textstyle \int_\bR \Re f = 0}\}''.
\end{equation}
Being invariant under $\g_{(\l,0)}$, this is again a proper subalgebra of $\cAm(O)$. However, as pointed out to us by D. Buchholz, this has the serious drawback that the resulting family of von Neumann algebras does not satisfy isotony: if $J \Subset I$ and $\supp f \subset J$ with $\int_\bR \Re f = 0$, it is easy to see that for sufficiently small $t >0$  $\supp \tau^{(m)}_t f \subset I$ but $\int_\bR \Re \tau^{(m)}_t f \neq 0$, and therefore $\am_t(W(f))$ does not belong to the algebra associated to $O_I$. This also shows that the union of all such algebras associated to double cones based on the time zero line is not invariant under time translations. 
\end{Remark}
\medskip

As shown in the following statement, the subnet $\cCoim$ captures the relevant information about the above described sectors of $\Aoim$.

\begin{Proposition}\label{prop:restrictionz}
Let $\r_{q,\iota}$ be a sector of the scaling limit theory $\Aoim$, then its
 restriction  to the scaling limit net $\cCoim$ is well defined and properly supported  and induces a non trivial translation covariant sector of $\cCoim$. 
 Moreover, the right cohomological extension of $\r_{q,\iota} \rest \cCoim$ coincides with $\r_{q,\iota}$.
\end{Proposition}

\begin{proof}
According to Thm. \ref{thm:limitCm}, 
every element  $A \in \cCoim(O_I)$ is a strong limit of linear combinations of Weyl operators $\Woi(f)$ with $\supp f \subset I$ and $\int_\bR f = 0$. 
For any such Weyl operator, 
by the Weyl relations $\r_{q,\iota}(\Woi(f)) = \Vqn^*\Woi(f)\Vqn$ differs from $\Woi(f)$ by a phase factor, and therefore $\r_{q,\iota}(A)$ still belongs to $\cCoim(O_{\tilde{I}})$, where $\tilde{I} \Supset I$. 
This shows that 
the restriction of $\r_{q,\iota}$ to $\cCoim$ is well defined. 
Moreover, such restriction is properly supported because $\r_{q,\iota}$ is, and
it is not equivalent to the vacuum sector  by a similar
argument as in~\cite[Sec. 4]{BV2}, when one observes that the operators $Z_w^{(n)}(\pi/q)$ used there, which form an asymptotically central sequence, belong to $\cCoim$ too, and $\r_{q,\iota}(Z_w^{(n)}(\pi/q)) = e^{-i \pi}Z_w^{(n)}(\pi/q)$ for $n$ large enough. Finally, the statements about translation covariance and the cohomological extension are consequences of the following observations. For $\bdx \in \bR$, the morphism
\[
\r^{(\bdx)}_{q,\iota}(A) = \lim_{n\to +\infty} \Woi(i\tau^{(0)}_\bdx u_n)^* A \Woi(i\tau^{(0)}_\bdx u_n), \qquad A \in \Aoim,
\]
is localized in $W_+ - a +\bdx$, restricts to a morphism of $\cCoim$ by the same argument used for $\r_{q,\iota}$, and is equivalent to the latter, a unitary intertwiner being
\[
W_\bdx = \lim_{n \to +\infty} \Woi(iu_n) \Woi(i\tau^{(0)}_\bdx u_n)^* ,
\]
where the limit exists in the strong operator topology as discussed in~\cite{BV2}. Moreover, for any wedge
$W \Supset (W_+-a) \cup (W_+ - a + \bdx)$ 
we have $\Woi(iu_n) \Woi(i\tau^{(0)}_\bdx u_n)^* \in \cCoim(W)$ for all $n \in \bN$ because $\int_\bR (u_n - \tau^{(0)}_\bdx u_n)=0$,
and thus $W_\bdx \in \cCoim(W)''$. 
Therefore, for each $A \in \Aoim(O)$, given $\bdx \in \bR$ such that $W_+ - a+\bdx \subset O'$, one has
\[
\r_{q,\iota}(A) = W_\bdx \rho_{q,\iota}^{(\bdx)}(A)W_\bdx^* = W_\bdx AW_\bdx^* \ ,
\]
as desired.
\end{proof}
As a matter of fact, the restriction of $\r_{q,\iota}$ to $\cCoim$ is localized in $O_{(-a,a)}$, cf. \cite[Prop. 4.5]{Cio}.

It is now possible to construct explicit examples of asymptotic morphisms of the net $\cCoim$, satisfying the general properties discussed in \cite{CM2}.

\begin{Theorem}\label{thm:asympmorCm}
The family $(\phi_\l)_{\l > 0}$ defined in Eq.~\eqref{eq:phil} is an asymptotic isomorphism of $\cCm$ with respect to $\uooi = \lim_\k \uo_{\l_\k}$.
Moreover, given a sector $\r_{q,\iota}$ of the scaling limit theory $\Aoim$, the family $(\rho_\l)_{\l > 0}$, defined by the norm limit
\begin{equation}\label{eq:rl}
\rho_\l(A) = \lim_n \uW (i u_n)_\l \phi_\l(A) \uW(iu_n)^*_\l \ ,\quad A \in \cCm \ , 
\end{equation}
 is a tame asymptotic morphism of $\cCm$
and it holds
\begin{equation}\label{eq:R2}
\poiac(\rac(A)) = \r_{q,\iota}\big(\poiac\pac(A)\big)= \r_{q,\iota}\big(\ad V^*(A)\big), \quad A \in \cCm.
\end{equation}
\end{Theorem}

\begin{proof}
We start by proving that $(\phi_\l)_\l$ is an asymptotic isomorphism. 
The properties (5.1)-(5.3) of~\cite{CM2} are obvious since $\phi_\l$ is an automorphism for each $\l > 0$. In order to prove properties (i) and (ii) of~\cite[def.\ 5.1]{CM2}, consider an element $A \in \cCm(O)$
and, given $\varepsilon > 0$, choose $h \in C_c(\bR^2)$ such that
\begin{equation}
\Big\| \Big[A-\int_{\bR^2} dx\,h(x)\az_x(A)\Big] \Oz\Big\| < \varepsilon.
\end{equation}
We have then the following equalities:
\begin{equation*}\begin{split}
\lim_\k \| [\phi_{\l_\k}(A) - \uam_{h}\pac(A)_{\l_\k}]\Oz\|^2 &= \lim_\k \| [A-\phi^{-1}_{\l_\k}(\uam_{h}\pac(A)_{\l_\k})]\Oz\|^2\\
&= \| [A-\ad V(\poi(\uam_{h}\pac(A)))]\Oz\|^2 \\
&= \Big\| \Big[A-\int_{\bR^2} dx\,h(x)\az_x(A)\Big] \Oz\Big\|^2 < \varepsilon^2.
\end{split}\end{equation*}
Here, the first equation follows from the fact that $\phi_\l$ is unitarily implemented and leaves the massless vacuum invariant, the second one follows from Thm.~\ref{thm:limitCm}\ref{it:defadV}, and the third one follows from Thm.~\ref{thm:limitCm}\ref{it:adVpoi} when one observes that $O$ is contained in some $O_I$ large enough. Since a similar argument holds for $\| [\phi_{\l_\k}(A) - \uam_{h}\pac(A)_{\l_\k}]^*\Oz\|$, 
we conclude that $\pac(A):=(\l \mapsto \phi_\l(A))$ belongs to $\uCCmac(O)$. Moreover, it is clear that the map $A \in \cCm_{\text{loc}} \mapsto \pac(A) \in {\uCCmac_{\text{loc}}}$ is norm continuous, and therefore it extends to a norm continuous map from $\cCm$ to $\uCCmac$. Since $\poiac ({\uCCmac_{\text{loc}}}) \subset \cC_{0,\iota,\text{loc}}^{(m)}$ by~\cite[thm.\ 4.6]{CM2}, this also shows that property (iii) of~\cite[def.\ 5.1]{CM2} is valid. It is also clear that, being $\phi_\l$ an automorphism, $A \mapsto \pac(A)$ is injective, i.e., we get property (i) of~\cite[def.\ 5.4]{CM2}. Finally we show the validity of property (ii) of~\cite[def.\ 5.4]{CM2}. To this end, we claim that the map $\bar s : \cCoim \to \uCCmac$ defined by
\[
\bar s(A_0)_\l := \phi_\l\ad V(A_0), \qquad A_0 \in \cCoim,
\]
is a continuous section of $\poiac : \uCCmac \to \cCoim$. The map is obviously continuous. Moreover, the fact that it is a section follows at once from the identity
\begin{equation}\label{eq:poiacpac}
\poiac(\pac(A)) = \ad V^*(A), \qquad A \in \cCm,
\end{equation}
applied to $A = \ad V(A_0) \in \cCz = \cCm$. In turn, the latter equation is proven through the equalities 
$$\poiac \pac (A) = \lim_{h \to \delta} \poi(\uam_h \pac(A)) = \lim_{h \to \delta} \ad V^* \bigg( \int_{\bR^2} dx\,h(x) \az_x(A) \bigg) = \ad V^* (A) \  $$
(where the limits are taken in the strong operator topology),
which are consequences of \cite[Lemma 4.5]{CM2} and Thm.\ \ref{thm:limitCm}(v).
The proof of (ii) of~\cite[def.\ 5.4]{CM2} is then achieved by observing that by Thm.~\ref{thm:limitCm}\ref{it:adVCoim} $\ad V(\bigcup_O \cCoim(O)) = \bigcup_O \cCz(O)= \bigcup_O \cCm(O)$, and therefore $\bar s(\bigcup_O \cCoim(O)) = \pac \ad V(\bigcup_O \cCoim(O)) = \pac(\bigcup_O \cCm(O))$.

In order to show that $\rho_\l$ is a tame asymptotic morphism, we recall that $\uW(f) \in \bigcup_O \uCCmac(O)$ for $f \in \cD(\bR)$ such that $\int_\bR f =0$. This fact, together with what we have just shown, makes it clear that $\rac(A) \in \uCCmac$ for all $A \in \cCm$ and that 
$\poiac(\rac(\bigcup_O \cCm(O))) \subset \bigcup_O \cCoim(O)$, as required. All the remaining properties are obviously satisfied.

Finally, the formula (\ref{eq:R2}) can be verified by a direct computation. Indeed, by~\eqref{eq:rl} and~\eqref{eq:rqi},
$$\poiac(\rac(A)) = \lim_n \Woi (iu_n) \poiac \pac (A) \Woi(iu_n)^* = \r_{q,\iota}(\poiac \pac (A))$$
and then one uses~\eqref{eq:poiacpac}.
\end{proof}
It can also be shown that $\cCm$ has convergent (and therefore unique) scaling limit, by repeating the argument in \cite[Thm.\ 7.5]{BDM2}, \emph{mutatis mutandis}.

\section{Asymptotic morphisms and smoothed out Weyl operators}\label{sec:Schwinger}

As remarked in Sec.~\ref{sect:Free}, $\cAm$ and $\Aoim$ are not isomorphic and therefore $(\phi_\l)$ can not be an asymptotic isomorphism of $\cAm$. It is however natural to ask how far it goes in this direction. To this end, a relevant condition is that the function $\l \mapsto \phi_\l(A)$ belongs to $\uAAmac$ for all $A \in \cAm$. We will see shortly that this is not the case in general. However, a partial result in this direction can be formulated introducing the local, Poincar\'e covariant net of C*-algebras $O \mapsto \fWm_r(O)$ generated by smoothed out Weyl  operators, i.e., by elements of $\cAm(O)$ of the form
\begin{equation}\label{eq:regweyl}
\int_{\bR^2} dy \,g(y) \am_y(W(f)), \qquad g \in C_c(\bR^2).
\end{equation}
It is easy to verify that  the action $x \mapsto \am_x(W)$ of translations is norm continuous for all $W \in \fWm_r(O)$ and, using 
\eqref{eq:innercont} and Lorentz covariance, that $\fWm_r(O)$ is strongly dense in $\cAm(O)$.

We also define the C*-subalgebra $\cAm_\phi(O) \subset \cAm(O)$ of all elements $A \in \cAm(O)$ which are mapped into $\uAAmac(O)$ by the isometric morphism $\pac : \cAm \to \ell^\infty(\bR_+,B(\cH))$ (the C*-algebra of bounded functions from $\bR_+$ to $B(\cH)$), and we denote by $\cAm_\phi$ the inductive limit of the net $O \mapsto \cAm_\phi(O)$.
We already know, by~\eqref{eq:weylreg}, that $W(f) \in \cAm_\phi(O_I)$ for any test function $f \in \cD(I)$. We now show that $\fWm_r(O) \subset \cAm_\phi(O)$.

\begin{Proposition}\label{prop:philW}
For all $W \in \fWm_r(O)$, $\l \mapsto \phi_\l(W)$ belongs to $\uAAmac(O)$.
\end{Proposition}

\begin{proof}
Thanks to the fact the $\pac$ is an isometric morphism, it is sufficient to prove the statement for elements $W \in \fWm_r(O)$ of the form~\eqref{eq:regweyl}. Then, for such a $W$ one has
\[\begin{split}
\sup_{\l \in (0,1)} \big\|\big[\phi_\l(W)-\am_{\l x}&\phi_\l(W)\big]\Om\big\|\\
&=\sup_{\l \in (0,1)}\left\| \int_{\bR^2}dy\,g(y)\big[\phi_\l(W(\tau_y^{(m)}f))-\am_{\l x}\phi_\l(W(\tau^{(m)}_yf))\big]\Om\right\| \\
&\leq \int_{\bR^2}dy\,|g(y)|\sup_{\l \in (0,1)}\big\|\big[\phi_\l(W(\tau_y^{(m)}f))-\am_{\l x}\phi_\l(W(\tau^{(m)}_yf))\big]\Om\big\|,
\end{split}\]
and therefore~\cite[Eq. (4.13)]{BV2}, together with the dominated convergence theorem, shows that 
\[
\lim_{x \to 0}\sup_{\l \in (0,1)} \big\|\big[\phi_\l(W)-\am_{\l x}\phi_\l(W)\big]\Om\big\| = 0.
\]
This, in turn, implies that $\limsup_{\l \to 0} \| [\phi_\l(W) - \ua_h (\pac(W))_\l]\Om\|$ can be made arbitrarily small for $h \in C_c(\bR^2)$ sufficiently close to a delta function. In a similar way, $\limsup_{\l \to 0} \| [\phi_\l(W) - \ua_h (\pac(W))_\l]^*\Om\|$ can be made small as well and therefore $\pac(W) \in \uAAmac(O)$, as desired.
\end{proof}

\begin{Proposition}\label{prop:notiso}
Let $W = \int_{\bR^2} dy \,g(y) \am_y(W(f))$ be such that $\int_\bR \Re f = 0$, $\int_\bR \Im f \neq 0$, and $g \in C_c(\bR^2)$ non-negative and not identically zero. Then $W \neq 0$ and $\poiac\pac(W) = 0$. 
\end{Proposition}

\begin{proof}
One has that
\[
\langle \Om, W\Om\rangle = \int_{\bR^2}dx\,g(x) e^{-\frac 1 2 \| \tau_x^{(m)} f\|_m^2}
\]
is strictly positive, and therefore $W \neq 0$. We can now  choose a sequence $(\l_n)_{n \in \bN} \subset \bR_+$ such that
\begin{equation}\label{eq:normzero}\begin{split}
\|\poiac\pac(W)\Ooi\|^2 &= \lim_{n \to +\infty} \| \phi_{\l_n}(W)\Om\|^2\\
&= \lim_{n \to +\infty} \left\| \int_{\bR^2} dx\, g(x) W(\d_{\l_n}\tau_x^{(m)}f)\Om\right\|^2 \\
&= \lim_{n \to +\infty} \int_{\bR^4} dxdy\, g(x)g(y) e^{\frac i 2 \sigma(\tau_x^{(m)}f,\tau_y^{(m)}f)}e^{-\frac 1 2 \| (\tau_y^{(m)}-\tau_x^{(m)})f\|^2_{\lambda_n m}},
\end{split}\end{equation}
where in the last equality we used the dilation invariance of the symplectic form. We now observe that, using the notation $x = (t,\bdx)$, $y = (s,\bdy)$, we have
\[\begin{split}
\int_{\bR} \Re(\tau_y^{(m)} - \tau_x^{(m)})f &= \int_{\bR} \Re(\tau_\bdy\tau_s^{(m)} - \tau_\bdx\tau_t^{(m)})f = \int_{\bR} \Re(\tau_s^{(m)} - \tau_t^{(m)})f\\
&= \big[\Re(\tau_s^{(m)} - \tau_t^{(m)})f\big]\widehat{\;}(0) = m[\sin(tm)-\sin(sm)]\widehat{\Im f}(0),
\end{split}\]
where we used the translation invariance of the integral in the second equality and~\eqref{eq:timeev} in the fourth one. The above quantity vanishes only if $s = t +\frac{2\pi}{m}\bZ$ or $s = -t +\frac{\pi}{m}(2\bZ+1)$, and therefore on a set of measure zero in $\bR^2$. Recalling then that, for $h \in \cD(\bR)$ such that $\int_\bR \Re h \neq 0$,
\[
\lim_{m \to 0} \|h\|_{m}^2 = \lim_{m\to 0} \int_{\bR} d\bdp\,\left|\frac{\widehat{\Re h}(\bdp)}{\sqrt{\omega_m(\bdp)}}+i\sqrt{\omega_m(\bdp)}\widehat{\Im h}(\bdp)\right|^2 = +\infty
\]
we see that the limit, as $n \to +\infty$, of the integrand in the last member of~\eqref{eq:normzero} vanishes almost everywhere, and therefore, by dominated convergence,
\[
\|\poiac\pac(W)\Ooi\|^2 = 0.
\]
The conclusion is then obtained by the separating property of $\Ooi$ for local algebras.
\end{proof}

We put on record a pair of immediate consequences of the above result.

\begin{Corollary}\label{cor:notsimple}
\renewcommand{\theenumi}{(\roman{enumi})}
\renewcommand{\labelenumi}{\theenumi}
\begin{enumerate} 
The following statements hold:
\item the quasi-local C*-algebras $\fWm_r$, $\cAm_\phi$ are not simple;
\item $\cAm_\phi$ is a proper subalgebra of $\cAm$.
\end{enumerate}
\end{Corollary}

The second statement in the corollary makes it plain that $\pac$ does not map $\cAm$ into $\uAAmac$.
However, as shown in the next proposition, the map $\poiac\pac : \cAm_\phi \to \Aoim$ acts, on suitable elements, in a way that closely resembles the isomorphism between the free scalar field net in $d \geq 3$ and its scaling limit built in~\cite{BV2}.

\begin{Proposition}\label{prop:Wupzero}
Let $f \in \cD(\bR)$, $h \in C_c(\bR^2)$. 
There holds:
\renewcommand{\theenumi}{(\roman{enumi})}
\renewcommand{\labelenumi}{\theenumi}
\begin{enumerate}
\item ${\displaystyle \poi(\uam_h\uW(f)) = \int_{\bR^2} dx\,h(x) \aoim_x(\Woi(f))};$
\item there exists
\begin{equation}\label{eq:Wupzero}
\lim_{\l \to 0} \phi^{-1}_\l(\uam_h\uW(f)_\l) = \int_{\bR^2}dx\,h(x) W(\tau^{(0)}_xf) =:W^{(0)}_{h,f}\in \cAm
\end{equation}
in the strong operator topology, where the integral in the r.h.s.\ is defined in the strong sense;
\item $\pac(W^{(0)}_{h,f}) \in \uAAmac$ and $\poiac(\pac(W^{(0)}_{h,f})) = \poi(\uam_h\uW(f))$;
\end{enumerate}
\end{Proposition}

\begin{proof}
(i) Let $g \in C_c(\bR^2)$. One has:
\[\begin{split}
\int_{\bR^2} dy\,g(y)\aoim_y\big(\poi(\uam_h\uW(f))\big) &= \int_{\bR^2}dy\,g(y)\poi\big(\uam_y\uam_h\uW(f)\big)\\
&=\poi\bigg(\int_{\bR^2}dy\,g(y) \uam_y\uam_h\uW(f)\bigg)\\
&= \poi\big(\uam_{g*h}\uW(f)\big)\\
&= \int_{\bR^2}dx\,h(x)\aoim_x\big(\poi(\uam_g\uW(f))\big),
\end{split}\]
where we used in the second equality the fact that the C*-algebra-valued function 
$$y \mapsto g(y) \uam_y\uam_h\uW(f),$$ 
being continuous and compactly supported, is Bochner-integrable, and in the fourth one the commutativity of the convolution product. The statement in then obtained by taking, in the strong operator topology, the limit $g \to \delta$ on both sides, and by recalling that in such a limit $\poi(\uam_g\uW(f))\to\Woi(f)$.

(ii) We postpone for a moment the proof that the integral in the r.h.s.\ is well defined.
Using the commutation relations between dilations and translations on $\cD(\bR)$, $\d_\l \tau^{(\l m)}_x = \tau^{(m)}_{\l x}\d_\l$, there holds:
\[
\bigg\|\bigg[\phi^{-1}_\l(\uam_h\uW(f)_\l)-\int_{\bR^2}dx\,h(x) W(\tau^{(0)}_xf)\bigg]\Om\bigg\| \leq \int_{\bR^2}dx\,|h(x)| \big\|\big[W(\tau^{(\l m)}_xf)-W(\tau^{(0)}_xf)\big]\Om\big\|.
\]
Moreover, if $x = (\tau, \bdx)$, since $\| \tau_\bdx g \|_m = \|g\|_m$ for all $m \geq 0$, 
\begin{equation*}
\| \tau^{(\l m)}_x f - \tau_x^{(0)}f \|_m^2 = \|\tau_t^{(\l m)}f-\tau^{(0)}_tf\|^2_m
\end{equation*}
can be expressed by an integral as in Eq.~\eqref{eq:normtranslation}, with the change $|\bdp| \to \o_m(\bdp)$. Therefore the same argument used there, with the same change in~\eqref{eq:estimate}, guaranteees that the dominated convergence theorem is applicable, thus yielding $\lim_{\l \to 0} \| \tau^{(\l m)}_x f - \tau_x^{(0)}f \|_m^2 =0$. This in turn implies 
\[
\lim_{\l \to 0} \big\|\big[W(\tau^{(\l m)}_xf)-W(\tau^{(0)}_xf)\big]\Om\big\| = 0,
\]
and a further application of the dominated convergence theorem, together with the fact that the vacuum is separating for the local algebras, gives the statement. Finally, we notice that in a similar way one can show that $\|\tau^{(0)}_{t'} f-\tau^{(0)}_{t}f\|_m \to 0$ as $t' \to t$, which, together with the fact that space translations are mass independent, implies that the function $x\in \bR^2 \mapsto W(\tau^{(0)}_xf)$ is strongly continuous, and therefore the integral on the right hand side of~\eqref{eq:Wupzero} is well defined in the strong topology.

(iii) Similarly to point (ii) above, one has, for all $\l > 0$,
\[
\big\|\big[\phi_\l(W^{(0)}_{h,f})-(\uam_h\uW(f))_\l\big]\Om\big\| \leq \int_{\bR^2}dx\,|h(x)|\big\|\big[W(\d_\l\tau^{(0)}_xf)-W(\d_\l\tau_{x}^{(\l m)}f)\big]\Om\big\|,
\]
and $\lim_{\l \to 0}\| \d_\l\tau^{(0)}_xf-\d_\l\tau_{x}^{(\l m)}f\|_m^2 = \lim_{\l \to 0} \| \tau^{(0)}_t f - \tau^{(\l m)}_t f\|_{\l m}^2 = 0$. This last statement follows by observing that $\| \tau^{(0)}_t f - \tau^{(\l m)}_t f\|_{\l m}^2$ is expressed again by an integral obtained from the one in~\eqref{eq:normtranslation} by the replacement $|\bdp| \to \o_{\l m}(\bdp)$. Moreover, this integral can be split as the sum of an integral over the region $|\bdp| \leq 1$ and one over the region $|\bdp| > 1$. The latter integral is seen to converge to zero, as $\lambda \to 0$,
since the bound \eqref{eq:estimate}
is integrable for $|\bdp| > 1$. To treat the possible divergence, for $\l \to 0$, of the integral over $|\bdp| \leq 1$, we observe that, for $\l \in (0,1)$, there hold the elementary bounds
\begin{align*}
\frac{|\cos(t\o_{\l m}(\bdp))-\cos(t|\bdp|)|}{\o_{\l m}(\bdp)^{1/2}} &\leq |t| \frac{\o_{\l m}(\bdp)-|\bdp|}{\o_{\l m}(\bdp)^{1/2}} \leq |t| \, \o_m(\bdp)^{1/2},\\
\frac{| \o_{\l m}(\bdp)\sin(t\o_{\l m}(\bdp))-|\bdp|\sin(t|\bdp|)|}{\o_{\l m}(\bdp)^{1/2}} &\leq   \frac{\o_{\l m}(\bdp)-|\bdp|}{\o_{\l m}(\bdp)^{1/2}} |\sin(t\o_{\l m}(\bdp))|\\
&+|\bdp|\frac{ |\sin(t\o_{\l m}(\bdp))-\sin(t|\bdp|)|}{\o_{\l m}(\bdp)^{1/2}} \leq 3 \, \o_m(\bdp)^{1/2},
\end{align*}
which show that the integrand is bounded by an integrable function of $\bdp$ uniformly for $\l \in (0,1)$, and therefore one can apply the dominated convergence theorem once more.
Therefore, one concludes that
\[
\lim_{\l \to 0}\big\|\big[\phi_\l(W^{(0)}_{h,f})-(\uam_h\uW(f))_\l\big]\Om\big\| = 0,
\]
which, together with $\phi_\l(W^{(0)}_{h,f})^* = \phi_\l(W^{(0)}_{\bar h,-f})$, implies that $\pac(W^{(0)}_{h,f}) \in \uAAmac$ and,
using~\cite[Lemma~4.4]{CM2} and the fact that $\Ooi$ is separating for the local algebras $\Aoim(O)$,
also that $\poiac(\pac(W^{(0)}_{h,f})) = \poi(\uam_h\uW(f))$.
\end{proof}

In particular, one can deduce from the proof of point (iii) above that $\aoim_x(\Woi(f)) =\Woi(\tau^{(0)}_x f)$ for all $x \in \bR^2$ and $f \in \cD(\bR)$. Moreover, by similar arguments, using the expressions for Lorentz transformation in~\cite[Eqs.\ (7.12)-(7.13)]{BDM2}, it is also possible to show that 
 $\lim_{\l \to 0} \| \tau^{(0)}_\L f - \tau^{(\l m)}_\L f\|_{\l m}^2 = 0$, which entails
$\aoim_\L(\Woi(f)) =\Woi(\tau^{(0)}_\L f)$ for all $\L$ in the Lorentz group. Thanks to this observation, we see that the net $\cAz$ is isomorphic to the subnet of $\Aoim$ generated by the Weyl operators $\Woi(f)$, $f \in \cD(\bR)$, with an isomorphism mapping $\pi^{(0)}(W(f))$ to $\Woi(f)$ which intertwines the respective Poincar\'e group actions. In particular, we can also think of $\cC$ as a covariant subnet of $\Aoim$. Upon this identification, it follows from statements (i) and (iii) of Prop.~\ref{prop:Wupzero} that
\begin{equation}\label{eq:WhfinC}
\int_\bR \Re f = 0 \quad \Rightarrow \quad \poiac(\pac(W^{(0)}_{h,f})) = \int_{\bR^2} dx\,h(x) \aoim_x(\Woi(f)) \in \cC.
\end{equation}

According to Thm.~\ref{thm:asympmorCm}
and the results in this section, the relations among the various subnets of $\cAm$ introduced so far can be summarized as
$$
\begin{array}{ccccc}
\fWm_r \cup \fWz_r & \subset & \cAm_\phi & \subsetneq & \cAm \\  
&& \rotatebox{90}{$\subsetneq$} &  & \\
&& \cCm&&
\end{array} 
$$
where $\fWz_r$ is the net of $C^*$-algebras generated by the operators  in \eqref{eq:Wupzero}.

For reference's sake, we note that if $W \in \fWm_r$ is an element of the form~\eqref{eq:regweyl}, then
\[
\phi_\l(W) = \int_{\bR^2}dy\,g(y) \a^{(\l^{-1}m)}_{\l y}(W(\d_\l f)).
\]

In the sequel, we focus on the net of $C^*$-algebras $\cAm_\phi$. By Prop.~\ref{prop:philW}, $\cAm_\phi$ is strongly locally dense in $\cAm$.
As above, we consider 
\[\rho_\lambda(B) = \rho(\l) \phi_\l (B) = \lim_n \uW(iu_n)^*_\l \phi_\l(B) \uW(iu_n)_\l, \qquad \l >0,\, B \in \cAm_\phi,
\]
where the limit exists in the norm topology, and $\rho_\l$ is a morphism from $\cAm_\phi$ into $\cAm$.

\begin{Proposition}
The following statements hold:
\renewcommand{\theenumi}{(\roman{enumi})}
\renewcommand{\labelenumi}{\theenumi}
\begin{enumerate}
\item  for $B \in \cAm_\phi$, one has $\rac(B) \in \uAAac$ and
\begin{equation}\label{eq:rhoqaphi}
\poiac(\rac(B)) = \r_{q,\iota}(\poiac\pac(B)) \in \Aoim;
\end{equation}
\item the map $B \in \cAm_\phi \mapsto \rac(B) \in \uAAac$ is norm continuous;
\item $\poiac\left(\rac\left(\bigcup_O \cAm_\phi(O)\right)\right) \subset \bigcup_O \Aoim(O)$.
\end{enumerate}
\end{Proposition}

\begin{proof}
(i) For a local element $B \in \cAm_\phi(O)$, both assertions follow from the fact that, for $n$ large enough (namely, $n$ such that $\bar O$ is in the left spacelike complement of the point $(0,na)$),
\[
\rho_\lambda(B) = \uW(iu_n)^*_\l \phi_\l(B) \uW(iu_n)_\l, \qquad \l > 0,
\]
and, since $\poiac\pac(B) \in \Aoirm(O)$ by~\cite[Thm.\ 4.6]{CM}, also
\[
\r_{q,\iota}(\poiac\pac(B)) = \Woi(iu_n)^*\poiac\pac(B)\Woi(iu_n).
\] 
The extension to quasi-local elements in $\cAm_\phi$ is then simply a consequence of the fact that $\rac$, $\poiac$, $\pac$ and $\r_{q,\iota}$ are C*-algebra morphisms.

(ii) This is an immediate consequence of the obvious fact that $\rac : \cAm_\phi \to \uAAac$ is a C*-algebra morphism.

(iii)  As seen in the proof of (i), $\poiac(\pac(B)) = \r_{q,\iota}(\poiac\pac(B))$ belongs to $\bigcup_O \Aoim(O)$ for $B \in \bigcup_O \cAm_\phi(O)$.
\end{proof}

By letting $h$ converge to a $\delta$ function in Eq.~\eqref{eq:WhfinC}, we observe
that $\poiac\pac(\cAm_\phi(O))^-$ contains $\cC(O)$
. We deduce that $\r_{q,\iota}$, being strongly continuous in restriction to local algebras, is uniquely determined  on $\cC$  by the knowledge of $(\r_\l)$ through Eq.~\eqref{eq:rhoqaphi}, and then on the whole $\Aoim$ by a cohomological extension procedure as in Prop.~\ref{prop:restrictionz}.
In particular, setting $B = W_{h,f}^{(0)}$ in~\eqref{eq:rhoqaphi} we obtain, using Prop.~\ref{prop:Wupzero}(i) and (iii), and Eq.~\eqref{eq:rhol},
\[
\r_{q,\iota}\bigg(\int_{\bR^2} dx\, h(x) \aoim_x(\Woi( f))\bigg) = \poiac\bigg(\l \mapsto \int_{\bR^2} dx\,h(x) e^{-i\int_\bR u_\infty \Re( \tau_x^{(0)}f)} W(\tau^{(0)}_x\d_\l f)\bigg).
\]

By Cor.~\ref{cor:notsimple} $\cAm_\phi$ is a proper subalgebra of $\cAm$. This implies in particular that we can not make sense of Eq.~\eqref{eq:rhoqaphi} for an arbitrary $B \in \cAm$.
Moreover, $(\phi_\l)$ is not an asymptotic isomorphism of $\cAm$ in the sense of \cite[Def. 5.4]{CM} and
therefore we can not establish the analogue of Eq.~(5.5) of~\cite{CM}. To obtain such a result, one could be tempted to enlarge the algebra $\uAAac$ so as to encompass all functions of the form $\l \mapsto \phi_\l(A)$, $A \in \cAm$, but  Prop.~\ref{prop:notiso} implies that this can not be done in such a way that $\poiac$ is multiplicative on the enlarged algebra. 
In a sense, one might think that the situation at hand hints at a not yet existing notion of {\it unbounded asymptotic morphism}.

The origin of these complications has to be ascribed to the bad infrared behaviour, indeed
we will show in Sec.~\ref{sec:highdim} that the above approach works well for the free charged field in $d+1$ dimensions, where $d=2,3$. 

\medskip
In the last part of this section, we discuss one more aspect of our setting that may be of independent interest.
The following result is a version, adapted to the Cauchy data formulation of the free field we are adopting here, of the calculations in~\cite[Sec.~4]{Buc1}.
For the sake of completeness we include a proof based on similar computations as those appearing in the proof of Prop.~\ref{prop:Wupzero}.

\begin{Lemma}\label{lem:limitweyl}
Given functions $f_1,\dots, f_n \in \cD(\bR)$ such that $\int_\bR \Re f_j = 0$, $j=1,\dots,n$, and functions $h_1,\dots,h_n \in C_c(\bR^2)$, $n \in \bN$, there exists
$$
\lim_{\l \to 0}\om\big(\uam_{h_1}\uW(f_1)_\l\dots\uam_{h_n}\uW(f_n)_\l\big)
$$
\end{Lemma}

\begin{proof}
Exploiting the commutation relations between translations and dilations, the Weyl relations, and the definition of the vacuum state, one has 
\begin{multline}\label{eq:omega}
\om\big(\uam_{h_1}\uW(f_1)_\l\dots\uam_{h_n}\uW(f_n)_\l\big) \\= \int_{\bR^{2n}}dx_1\dots dx_n h_1(x_1)\dots h_n(x_n)\eta_\lambda(x_1,\dots,x_n)\exp\Big\{-\frac{1}{2}\| \tau_{x_1}^{(\lambda m)}f_1 + \dots + \tau_{x_n}^{(\lambda m)}f_n\|_{\lambda m}^2\Big\},
\end{multline}
where
$$\eta_\lambda(x_1,\dots,x_n) = \exp\Big\{-\frac{i}{2}\sum_{1\leq i<j\leq n}\sigma(\tau_{x_i}^{(\lambda m)}f_i,\tau_{x_j}^{(\lambda m)}f_j)\Big\}. $$
Moreover, there holds $\tau^{(m)}_{(t,\bdx)} = \tau^{(m)}_t\tau_\bdx$, and therefore one has, by the definition of the action of time translations, and setting $g_j := \tau_{\bdx_j}f_j$
\begin{equation*}
\| \tau_{x_1}^{(\lambda m)}f_1 + \dots + \tau_{x_n}^{(\lambda m)}f_n\|_{\lambda m}^2 \\
= \frac{1}{2}\int_\bR d\bdp\bigg|\sum_{j=1}^ne^{it_j\omega_{\lambda m}(\bdp)}\Big[\frac{\widehat{\Re g_j}(\bdp)}{\sqrt{\omega_{\lambda m}(\bdp)}} + i\sqrt{\omega_{\lambda m}(\bdp)}\widehat{\Im g_j}(\bdp)\Big]\bigg|^2 \ .
\end{equation*}
The integrand in the last expression converges pointwise, as $\l \to 0$, to the corresponding value for $\lambda = 0$, and the following bounds hold uniformly for $\lambda\in[0,1]$:
$$ \frac{|\widehat{\Re g_j}(\bdp)|}{\sqrt{\omega_{\lambda m}(\bdp)}} \leq \frac{|\widehat{\Re g_j}(\bdp)|}{|\bdp|}, \qquad  \sqrt{\omega_{\lambda m}(\bdp)}|\widehat{\Im g_j}(\bdp)| \leq \sqrt{\omega_{m}(\bdp)}|\widehat{\Im g_j}(\bdp)|.$$
Since $\widehat{\Re g_j}$, $\widehat{\Im g_j}$ are Schwartz functions, and $\widehat{\Re g_j}(0) = 0$, the two functions on the right hand sides of the above inequalities are square-integrable, and therefore, by the dominated convergence theorem,
$$\lim_{\lambda \to 0} \| \tau_{x_1}^{(\lambda m)}f_1 + \dots + \tau_{x_n}^{(\lambda m)}f_n\|_{\lambda m}^2 = \| \tau_{x_1}^{(0)}f_1 + \dots + \tau_{x_n}^{(0)}f_n\|_{0}^2.$$

We now consider the limit of $\eta_\lambda$. There holds, with the same notations as above, 
\begin{multline*}
\sigma(\tau_{x_i}^{(\lambda m)}f_i,\tau_{x_j}^{(\lambda m)}f_j) \\
= \int_\bR d\bdp\bigg[\sin\big[(t_j-t_i)\omega_{\lambda m}(\bdp)\big]\bigg(\frac{\widehat{\Re g_i}(\bdp)\widehat{\Re g_j}(\bdp)}{\omega_{\l m}(\bdp)}+\o_{\l m}(\bdp)\widehat{\Im g_i}(\bdp)\widehat{\Im g_j}(\bdp)\bigg)\\
+ \cos\big[(t_j-t_i)\o_{\l m}(\bdp)\big]\Big(\widehat{\Re g_i}(\bdp)\widehat{\Im g_j}(\bdp)-\widehat{\Im g_i}(\bdp)\widehat{\Re g_j}(\bdp)\Big)\bigg],
\end{multline*}
and, by a similar argument as before employing the dominated convergence theorem, one sees that
$$\lim_{\lambda \to 0} \eta_\lambda(x_1,\dots,x_n) = \eta_0(x_1,\dots, x_n).$$

Then, we have shown that the integrand in~\eqref{eq:omega} converges to the corresponding value for $\l = 0$, and thanks to the fact that $h_j \in C_c(\bR^2)$ and that the other factors are bounded uniformly in $\l$, we obtain the thesis, appealing again to the dominated convergence theorem.
\end{proof}

\begin{Proposition}\label{prop:Csubnet}
For all double cones $O \subset \bR^2$, 
\begin{equation}
\cC(O) \subset \poiac(\uAAmac(O)).
\end{equation}
\end{Proposition}

\begin{proof}
By Poincar\'e covariance, it is sufficient to prove the statement for $O = O_I$, a double cone based on the time zero line. Moreover, following the argument in~\cite[Thm. 4.8]{CM2}, it is sufficient to show that there is a C*-subalgebra $\uBB(O) \subset \uAAm(O)$ with the properties that for all $\uB \in \uBB(O)$ there exists $\lim_{\l \to 0} \om(\uB_\l) = \ooim(\uB)$, and that $\poi(\uBB(O))$ is dense in $\cC(O)$ in the strong* operator topology. We take as $\uBB(O)$ the C*-subalgebra of $\uAAm(O)$ generated by the operators $\uam_h \uW(f) \in \uAAm(O)$ with $\int_\bR \Re f = 0$. Then, from Lemma~\ref{lem:limitweyl} it follows at once, by an $\varepsilon/3$-argument, that actually $\lim_{\l \to 0} \om(\uB_\l) = \ooim(\uB)$ for all $\uB \in \uBB(O)$. The fact that $\poi(\uBB(O))$ is strongly* dense in $\cC(O)$ follows from the strong* limit $\poi(\ua_h\uW(f)) \to \Woi(f)$ for $h \to \delta$, see Eq.~\eqref{eq:Woi}.\footnote{Note that if $\supp f\subset O$, since $O$ is open it is always possible to choose $\supp h$ so small that $\ua_h\uW(f) \in \uAAm(O)$.} 
\end{proof}
This property should be compared to \cite[Thm.\ 4.8]{CM2} and could replace it in the formulation of a notion of asymptotic morphism relative to a suitable subnet of the scaling limit net for theories without convergent scaling limit. We plan to address this issue elsewhere.

\section{Asymptotic morphisms for the free charged scalar field} \label{sec:highdim}

Let $\varphi$ be the mass $m \geq 0$ free charged scalar field in $d = 4$ spacetime dimensions, and let $O \mapsto \cFm(O)$, resp.\ $O \mapsto \cAm(O) = \cFm(O)^{U(1)}$, be the corresponding field, resp.\ observable, net of von Neumann algebras, in the locally Fock representation induced by the massless vacuum state, as in~\cite{BV2}. We recall that, in such representation, $\cFm(O) = \cFz(O)$, $\cAm(O) = \cAz(O)$ for all double cones $O$ with base in the time-zero hyperplane. This entails, in particular, that the corresponding quasi-local field and observable C*-algebras coincide for all possible values of $m \geq 0$. Moreover, we recall that $\cFz$ is covariant under an automorphic action of the dilation group, denoted by $\l \in \bR_+ \mapsto \phi_\l \in \aut(\cFz)$.

The superselection structure of $\cAm$ is well known: the sectors are in 1-1 correspondence with $\bZ$, they are all simple, and the representative automorphism $\gamma^{(m)}_n$, $n \in \bZ$, localized in $O$, can be expressed by~\cite{F}, \cite[Sec.\ 8.4.B]{BLOT}
\begin{equation}\label{eq:gamman}
\gamma^{(m)}_n(A) = \psi_f^n A (\psi_f^n)^*, \qquad A \in \cAm,
\end{equation}
where $\psi_f \in \cFm(O)$ is the unitary phase of the polar decomposition of $\varphi(f)$, with $f \in \cD_\bR(O)$ such that its Fourier transform is not identically zero on the mass $m$ hyperboloid.

On the other hand, according to~\cite{DM, CM}, the outer regularized scaling limit nets $\Foirm, \Aoirm$ can be identified with the corresponding nets $\cFz, \cAz$ through the net isomorphism $\phi : \Foirm \to \cFz$ defined by
\begin{equation}\label{eq:isoscalinglimit}
\phi(\poi(\uF)) = \wlim_{\k} \phi_{\l_\k}^{-1}(\uF_{\l_\k}), \qquad \uF \in \uFF^{(m)},
\end{equation}
which is unitarily implemented. This also implies that $\Aoirm$ is a Haag-dual net. Moreover, since $\Aoim$ satisfies essential duality by~\cite[Prop.\ 6.3]{BV}, one has that the outer regularized net $\Aoirm$ coincides with the dual net $\cA_{0,\iota}^{(m)d}$. Therefore, in particular, combining~\eqref{eq:gamman} and~\eqref{eq:isoscalinglimit}, the superselection structure of the scaling limit net $\Aoim$ coincides with that of $\cAz$ described above. We also recall that, as quasi-local C*-algebras, $\Aoirm = \Aoim$.

\begin{Proposition}\label{prop:R4}
Let $f \in \cD_\bR(O)$, with $O$ a double cone with base in the time-zero hyperplane, be such that its Fourier transform does not vanish identically on the mass $m$ hyperboloid. With the notation $f_\l(x) := \l^{-3} f(\l^{-1}x)$, $x \in \bR^4$, the formula
\[
\rho_{n,\l}(A) := \psi_{f_\l}^n \phi_\l(A) (\psi_{f_\l}^n)^*, \qquad \l>0, A \in \cAm,
\]
defines a tame asymptotic morphism of $\cAm$ in the sense of~\cite[Def.\ 5.1]{CM2} for every $n \in \bZ$, such that $\poiac\br^\bullet_n = \phi^{-1} \gamma^{(0)}_n$, as an (iso)morphism from the quasi-local C*-algebra $\cAm = \cAz$ to $\Aoim$.
\end{Proposition}

\begin{proof}
The chosen representation of the nets generated by the free charged scalar fields of different masses is such that, for $f$ as in the statement, the field operators $\varphi(f)$ all coincide. This implies that the dilations of the mass $m = 0$ theory act on $\varphi(f)$ according to
\[
\phi_\l(\varphi(f)) = \varphi(f_\l), \qquad \l > 0,
\]
which entails, thanks to the unicity of the polar decomposition, $\phi_\l(\psi_f) = \psi_{f_\l}$. In turn, this has as a consequence that, according to~\cite[Thm.\ 4.4]{DM}, the sectors of $\cAm$ are all preserved, and that all the functions $\l \mapsto \psi_{f_\l}^n$, $n \in \bZ$, belong to ${\uFF^{(m)}}^\bullet(O)$. Moreover, $\ad \,\psi_{f_\l}^n$ is precisely the morphism $\rho(\l)$ of $\cAm$ appearing in~\cite[Prop.\ 7.1]{CM2}.

Furthermore, we claim that $(\phi_\l)$ is an asymptotic isomorphism of $\cAm$ in the sense of~\cite[Def.\ 5.4]{CM2} and that
\begin{equation}\label{eq:poiacbd}
\poiac(\bp^\bullet(A)) = \phi^{-1}(A), \qquad A \in \cAm.
\end{equation}
The validity of such statements is proven by arguments similar to those in the proof of Thm.~\ref{thm:asympmorCm}, using the properties of the isomorphism $\phi : \Foirm \to \cFz$ obtained in~\cite{BV2}.

We conclude that, by~\cite[Thm.\ 7.2]{CM2}, the $\r_{n,\l}$ as in the statement is a tame asymptotic morphism. Moreover, again by~\cite[Thm.\ 7.2]{CM2}, and by the definition of $\Psi$ in~\cite[Eq. (5.5)]{CM2}, we have
\[\begin{split}
\poiac \br^\bullet_n &= \Psi\big((\r_{n,\l}), (\phi_\l)\big)  \phi^{-1} = \ad\, \big[\poiac(\pac(\psi_f^n))\big]\phi^{-1} \\
&= \ad\,\big[\phi^{-1}(\psi_f^n)\big]\phi^{-1} = \phi^{-1}\ad\,\psi_f^n = \phi^{-1}\gamma^{(0)}_n,
 \end{split}\]
where in the third equality~\eqref{eq:poiacbd} was used. This completes the proof.
\end{proof}

Notice that a similar argument applies to the charged free scalar field in $d=3$.

\medskip

Let us now consider the tensor product theory $\cA^{(m_1)}\otimes \cA^{(m_2)}$. There is a (tensor product) net isomorphism, which we still denote by $\phi$, between the corresponding outer regularized scaling limit theory and the tensor product $\cAz \otimes \cAz$~\cite[Thm.~3.8]{DM}, whose superselection sectors are represented by the tensor product automorphisms $\gamma^{(0)}_h \otimes \gamma^{(0)}_k$, $(h,k) \in \bZ^2$. Considering then $\r_{h,k,\l} := \r_{h,\l}\otimes \r_{\k,\l}$ one has that $(\r_{h,k,\l})$ is a tame asymptotic morphism of $\cA^{(m_1)}\otimes \cA^{(m_2)}$, and $\poiac \br^\bullet_{h,k} = \phi^{-1}\gamma^{(0)}_h\otimes \gamma^{(0)}_k$. Notice that the existence of $\r_{h,k,\l}$ is due to the fact that  $\r_{h,\l}$ and $\r_{\k,\l}$ are \emph{bona fide} morphisms for all $\l > 0$. In more general situations, because of the fact that asymptotic morphisms are not necessarily linear maps, it is not obvious how to define a kind of tensor product.

\appendix

\section{On the quasiequivalence of vacuum states of different masses in 2d}\label{app:quasiequiv}
Recall that in Sec.~\ref{sect:asymptotic} we defined the von Neumann algebras
\[
\cCm(O_I) := \Big\{ \pi^{(m)}(W(f))\,:\, {\rm supp}\, f \subset I, \, {\textstyle \int_I f} = 0\Big\}'', \quad m\geq 0 \ .
\]
In this appendix we provide a proof of the following result.

\begin{Theorem}\label{thm:quasiequiv}
For each bounded interval $I \subset \bR$ there exists a von Neumann algebra isomorphism $\phi$ between $\cCm(O_I)$, $m>0$, and 
$\cCz(O_I)$ such that $\phi(\pi^{(m)}(W(f))) = \pi^{(0)}(W(f))$.
\end{Theorem}

To show this fact, we make appeal to the general results of~\cite{AY}. To make contact with the formalism employed there in the description of CCR algebras, we introduce then the notation
\[
\cD_0(I) := \Big\{ f \in \cD(I)\,:\, {\textstyle \int_I f } = 0\Big\},
\]
considered as a real vector space, and we define the complex vector space $K := \cD_0(I) \oplus \cD_0(I)$ with complex structure defined by
\[
i \cdot (f_1\oplus f_2) := (-f_2)\oplus f_1.
\]
Moreover, we define a conjugate linear involution $\Gamma : K \to K$ and an indefinite inner product $\gamma : K \times K \to \bC$ as
\begin{align*}
\Gamma(f_1 \oplus f_2) &:= f_1\oplus(-f_2),\\
\gamma(f_1 \oplus f_2,g_1\oplus g_2) &:= \frac 1 2\big(\langle f_1+if_2,g_1+ig_2\rangle - \langle g_1-ig_2,f_1-if_2\rangle\big),
\end{align*}
with $\langle \cdot, \cdot\rangle$ the standard scalar product on $L^2(\bR)$. They satisfy $\gamma(\Gamma x, \Gamma y) = -\gamma(y,x)$ for all $x, y \in K$. To these data,~\cite[Sec.\ 2]{AS} associates the self-dual CCR algebra $\fA(K,\gamma,\Gamma)$. This is the $*$-algebra generated by symbols $B(f_1 \oplus f_2)$, $f_1 \oplus f_2 \in K$, satisfying the natural relations suggested by the identification of $B(f_1 \oplus f_2)$ with the Fock space operators $\varphi(f_1) + i \varphi(f_2)$, 
where $\varphi(f) = \phi(\Re f) - \pi(\Im f) = 
\frac 1 {\sqrt 2} \big[ a(\omega_m^{-1/2} \Re f + i \omega_m^{1/2} \Im f) + a(\omega_m^{-1/2} \Re f + i \omega_m^{1/2} \Im f)^* \big]$.
Notice that $\pi^{(m)}(W(f)) = e^{i \varphi(f)}$ for $m>0$, while $\pi^{(0)}(W(f)) = e^{i \varphi(f)} \otimes \Id$ acts on the tensor product of the (Bosonic) Fock space on $\cD_0(\bR)$ with a nonseparable multiplicity space.

We now introduce the $\bR$-bilinear map $\langle \cdot, \cdot\rangle_m : \cD_0(\bR) \times \cD_0(\bR) \to \bC$, $m \geq 0$, defined by
\begin{equation*}\begin{split}
\langle f,g\rangle_m := \frac 1 2 \int_\bR d\bdp &\overline{\left[\o_m(\bdp)^{-1/2} (\Re f)\hat{\,}(\bdp) + i\o_m(\bdp)^{1/2} (\Im f)\hat{\,}(\bdp)\right]}  \\
&\times \big[\o_m(\bdp)^{-1/2} (\Re g)\hat{\,}(\bdp) + i\o_m(\bdp)^{1/2} (\Im g)\hat{\,}(\bdp)\big]\,.
\end{split}\end{equation*}
It is clear that $\Re \langle \cdot, \cdot \rangle_m$ is a real scalar product inducing the norm $\|\cdot\|_m$  on $\cD_0(\bR)$ and it is easy to verify that $\Im \langle f, g \rangle_m = \frac 1 2 \Im\langle f, g \rangle$. We can then introduce on $ K \times K$ the hermitian form
\[
S_m(f_1\oplus f_2,g_1\oplus g_2)  := 
\langle f_1,g_1\rangle_m + i \langle f_1,g_2\rangle_m -i \langle f_2,g_1\rangle_m + \langle f_2,g_2\rangle_m.
\]
This is nothing but $\big\langle\Om,(\varphi(f_1) + i \varphi(f_2))^* (\varphi(g_1) + i \varphi(g_2))\Om\big\rangle $ if $m>0$
and $\big\langle\Oz,\big[(\varphi(f_1) + i \varphi(f_2))^* (\varphi(g_1) + i \varphi(g_2))\big] \otimes \Id \, \Oz\big\rangle $ otherwise.

One verifies that $S_m(x,x) \geq 0$ and $S_m(x,y) - S_m(\Gamma y,\Gamma x) = \gamma(x,y)$ for all $x,y \in K$, and then, by~\cite[Lem.\ 3.5]{AS}, there exists a unique quasi-free state $\varphi_m= \varphi_{S_m}$ on $\fA(K,\gamma,\Gamma)$ such that $\varphi_m(x^*y) = S_m(x,y)$.

We now observe that $\Re K := \{x \in K\,:\, \Gamma x = x\} = \cD_0(I)\oplus\{0\}$, and that for $f,g \in \cD_0(I)$ we get
\begin{equation}\label{eq:sympnorm}
\gamma(f\oplus 0, g\oplus 0) = i \, \Im \langle f,g\rangle, \qquad S_m(f\oplus 0,f\oplus 0) =  \|f\|_m^2.
\end{equation}
Then, by~\cite[Prop.\ 3.4]{AY} and by the separating property of the vacuum vector for $\cCm(O_I)$, the Weyl operators $W_{S_m}(x)$, $x \in \Re K$, in the GNS representation of $\fA(K,\gamma,\Gamma)$ induced by $\varphi_m$ generate a von Neumann algebra isomorphic to $\cCm(O_I)$ for all $m \geq 0$.

The next step is to define the scalar product
\[
\langle x,y\rangle_{S_m} := S_m(x,y) + S_m(\Gamma y, \Gamma x), \qquad x,y \in K,
\]
and to observe that
\begin{equation}\label{eq:normSm}
\langle f_1\oplus f_2,f_1\oplus f_2\rangle_{S_m} = 2 (\|f_1\|_m^2 + \|f_2\|_m^2) = 0 
\end{equation}
implies $f_1 \oplus f_2 = 0$. Then the main Theorem of~\cite{AY} implies that the statement of Thm.~\ref{thm:quasiequiv} holds if and only if:
\begin{enumerate}
\item the (Hausdorff) topologies induced on $K$ by $\langle \cdot, \cdot\rangle_{S_m}$ and $\langle \cdot, \cdot\rangle_{S_0}$ coincide;
\item denoting by $\bar K$ the completion of $K$ in the above topology, and by $\tilde S_m, \tilde S_0 : \bar K \to \bar K$ the bounded (by~\cite[Lem.\ 3.2]{AY}) operators defined by
\begin{equation}\label{eq:SmS0}
S_m(x,y) = \langle x, \tilde S_m y\rangle_{S_m}, \quad S_0(x,y) = \langle x, \tilde S_0 y\rangle_{S_m}, \qquad x,y \in K,
\end{equation}
(with $\langle \cdot,\cdot \rangle_{S_m}$ extended by continuity to $\bar K$), the operator $\tilde S_m^{1/2}- \tilde S_0^{1/2}$ is Hilbert-Schmidt.
\end{enumerate}
The remaining part of this appendix is devoted to showing the validity of 1. and 2. 
\begin{Proposition}
The topologies induced on $K$ by $\langle \cdot, \cdot\rangle_{S_m}$ and $\langle \cdot, \cdot\rangle_{S_0}$ coincide.
\end{Proposition}

\begin{proof}
By~\eqref{eq:normSm}, it is sufficient to show that the norms $\|\cdot\|_m$, $\|\cdot\|_0$ are equivalent on $\cD_0(I)$. Now, by the Parseval identity and by the fact that $\o_m^{\pm 1/2}$ maps real functions to real functions, 
\[ \begin{split}
\|f\|_m^2 &= \frac 1 2 \int_{\bR} d\bdx \left|\left(\o_m^{-1/2} \Re f + i \o_m^{1/2}\Im f\right)(\bdx)\right|^2
\\
&=\frac 1 2 \int_\bR d\bdp\left[ \o_m(\bdp)^{-1} |(\Re f)\hat{}(\bdp)|^2+\o_m(\bdp) |(\Im f)\hat{}(\bdp)|^2\right],
\end{split}\]
and therefore it is sufficient to show that the $\pm 1$ norms
\[
\| f\|_{m,\pm 1}^2 := \int_\bR d\bdp\,\o_m(\bdp)^{\pm1} |\hat{f}(\bdp)|^2, \qquad f \in \cD_0(I,\bR),
\]
are respectively equivalent for $m \neq 0$ and $m = 0$.

The inequalities $\|f\|_{m,-1} \leq \|f\|_{0,-1}$, $\|f\|_{0,1} \leq \|f\|_{m,1}$ are obvious. We now prove that $\|f\|_{m,1} \leq C \|f\|_{0,1}$ for some $C > 0$. As in~\cite{FG}, by the inequality
\begin{equation}\label{eq:omineq}
\o_m(\bdp) \leq \sqrt{1+m^2}|\bdp| + m \chi_{\{|\bdp|\leq 1\}}(\bdp)
\end{equation}
we get
\[
\|f\|^2_{m,1} \leq \sqrt{1+m^2} \|f\|^2_{0,1} +m \int_\bR d\bdp |\hat f(\bdp)|^2,
\]
and the required inequality is obtained by applying~\cite[Thm.\ 1]{Ma} with $n=1$, $p=2$, $\Omega = I$. Indeed, since $t \leq e^{\pi t}$ for all $t \geq 0$,
\[\begin{split}
 \int_\bR d\bdp |\hat f(\bdp)|^2 &= \int_I d\bdx |f(\bdx)|^2 \leq \|(-\Delta)^{1/4} f\|^2_{L^2(I)} \int_I d\bdx\, e^{\pi \Big|\frac{f(\bdx)}{\|(-\Delta)^{1/4} f\|_{L^2(I)}}\Big|^2} \\
 &\leq c_{1,2} |I| \,\|(-\Delta)^{1/4} f\|^2_{L^2(I)},
\end{split}\]
and moreover there exist $c' > 0$ such that $\|(-\Delta)^{1/4} f\|_{L^2(I)} \leq c' \|(-\Delta)^{1/4} f\|_{L^2(\bR)} = c' \|f\|_{0,1}$, see~\cite[Rem.\ 2]{Ma}. Summing up, we obtain $\|f\|_{m,1}^2 \leq (\sqrt{1+m^2}+m c_{1,2}\, c'^2\, |I|)\|f\|_{0,1}^2$.

Finally, we prove $\|f\|_{0,-1} \leq C \| f\|_{m,-1}$ for some $C > 0$. Using~\eqref{eq:omineq} again, we have
\[\begin{split}
\|f\|_{0,-1}^2 &\leq \sqrt{1+m^2} \int_{|\bdp| > 1} \frac{d\bdp}{\o_m(\bdp)} |\hat f(\bdp)|^2 +  \int_{|\bdp| \leq 1} \frac{d\bdp}{|\bdp|} |\hat f(\bdp)|^2 \\
&\leq \sqrt{1+m^2} \|f\|_{m,-1}^2 +  \int_{|\bdp| \leq 1} \frac{d\bdp}{|\bdp|} |\hat f(\bdp)|^2.
\end{split}\]
If we now consider a function $\varphi \in \cD(\bR)$ such that $\varphi(\bdx) = 1$ for all $\bdx \in I$, we have $f = f \varphi$, and therefore
\[\begin{split}
 \int_{|\bdp| \leq 1} \frac{d\bdp}{|\bdp|} |\hat f(\bdp)|^2 &=  \int_{|\bdp| \leq 1} \frac{d\bdp}{|\bdp|} \left| \int_\bR d\bdq \hat f(\bdq)\hat \varphi(\bdp - \bdq) \right|^2 = \int_{|\bdp| \leq 1} d\bdp |\bdp| \left| \int_\bR d\bdq \hat f(\bdq)\psi(\bdp,\bdq) \right|^2 \\
 &= \int_{|\bdp| \leq 1} d\bdp |\bdp| \left| \int_\bR d\bdq \frac{\hat f(\bdq)}{\o_m(\bdq)^{1/2}}\o_m(\bdq)^{1/2}\psi(\bdp,\bdq) \right|^2
\end{split}\]
where in the second equality we have introduced the function $\psi(\bdp,\bdq) := \frac{\hat \varphi(\bdp-\bdq)-\hat \varphi(-\bdq)}{\bdp}$ and used the fact that $\int_\bR d\bdq  \hat f(\bdq) \hat \varphi(-\bdq) = \hat f(0) = 0$. From this, the required estimate will follow by an application of the Cauchy-Schwarz inequality, provided we can show that
\begin{equation}\label{eq:finite}
C_1(m) := \int_{|\bdp| \leq 1} d\bdp |\bdp| \int_\bR d\bdq \,\o_m(\bdq)|\psi(\bdp,\bdq) |^2 < +\infty.
\end{equation}
To this end, we write $\psi(\bdp,\bdq) = \hat \varphi'(-\bdq) + \int_0^\bdp d\bds\, \hat \varphi''(\bds-\bdq)(\bdp-\bds)$ and we recall that $\hat \varphi$ is a Schwartz function, which entails the existence of a constant $K > 0$ such that $|\hat \varphi''(\bds)| \leq K/(1+|\bds|)^2$. Using then the fact that
\[
\int_{-\bdq}^{\bdp-\bdq} \frac{d\bds}{(1+|\bds|)^2} = \frac{1}{1+|\bdq|}-\frac{1}{1+|\bdp-\bdq|}
\]
for $|\bdq| > 1$, $|\bdp| \leq 1$, we obtain the estimate
\begin{multline*}
\int_{|\bdp| \leq 1} d\bdp |\bdp| \int_\bR d\bdq \,\o_m(\bdq)\left| \int_0^\bdp d\bds\, \hat \varphi''(\bds-\bdq)(\bdp-\bds)\right|^2 \\
\leq 4K^2 \int_{|\bdp| \leq 1} d\bdp |\bdp|^3 \left[|\bdp|^2\int_{|\bdq| \leq 1} d\bdq \,\o_m(\bdq) + \int_{|\bdq| > 1} d\bdq \,\o_m(\bdq)\left( \frac{1}{1+|\bdq|}-\frac{1}{1+|\bdp-\bdq|}\right)^2\right] < +\infty.
\end{multline*}
This, together with the obvious estimate $\int_{|\bdp| \leq 1} d\bdp |\bdp| \int_\bR d\bdq \,\o_m(\bdq)|\hat \varphi'(-\bdq)|^2 < +\infty$, shows, by the Minkowski inequality, the validity of~\eqref{eq:finite}. All in all, we obtain $\|f\|_{0,-1} \leq C \| f\|_{m,-1}$ with $C = [\sqrt{1+m^2}+C_1(m)]^{1/2}$.
\end{proof}

We stress that, thanks to a general result by Buchholz \cite[Appendix B]{Bu}, the condition 2. above is satisfied if $\tilde S_m - \tilde S_0$ is a trace-class operator on $\bar K$. In order to get a more explicit expression for these operators, we observe that, thanks to~\eqref{eq:sympnorm} and~\cite[Lem. 3.2]{AY}, there holds, for all $m \geq 0$, the estimate
\[
|\Im \langle f,g\rangle|^2 \leq 4 \|f\|^2_m \|g\|^2_m, \qquad f,g \in \cD_0(I),
\]
and therefore the symplectic form $\Im \langle \cdot,\cdot\rangle$ extends by continuity to the real Hilbert space $(\bDoIm, \Re\langle \cdot,\cdot\rangle_m)$, and there exist bounded antisymmetric operators $R_m : \bDoIm \to \bDoIm$, $m \geq 0$, such that
\[
\frac 1 2 \Im \langle f,g\rangle = \Re \langle f,R_m g\rangle_m, \qquad f,g \in \bDoIm.
\]

It is then a straigthforward computation to verify that~\eqref{eq:SmS0} is satisfied if
\[
\tilde S_m = \frac 1 2 \left(\begin{matrix} \Id &-R_m \\ R_m & \Id\end{matrix}\right), \quad \tilde S_0 = \frac 1 2  \left(\begin{matrix} T &-R_m \\ R_m & T\end{matrix}\right),
\]
with $T : \bDoIm \to \bDoIm$ the (symmetric) bounded (by Prop.~\ref{eq:normSm}) operator defined by
\[
\Re \langle f, Tg\rangle_m = \Re \langle f,g\rangle_0, \qquad f,g \in \bDoIm.
\]
Therefore the condition that $\tilde S_m - \tilde S_0$ is trace class is equivalent to the condition that $\Id - T$ is trace class on $\bDoIm$.

We also observe, in passing, that by~\cite[Lem.\ 3.3]{Ve} the factoriality of $\cCm(O_I)$ is equivalent to the fact that $R_m$ or $R_0$ (and then both) is injective with dense range. (Moreover, by antisymmetry, $R_m$ is injective if and only if it has dense range.) If this is the case, it is also easy to check that $T = R_m R_0^{-1}$.

We also put on record the following formula for the matrix elements of $\Id - T$:
\begin{multline}\label{eq:1-T}
\Re \langle f, (\Id-T)g\rangle_m = \Re \langle f,g\rangle_m - \Re \langle f,g\rangle_0 \\
= \frac 1 2\int d\bdp\,\left[\left(\frac 1 {\o_m(\bdp)}-\frac 1 {|\bdp|}\right)\overline{(\Re f)\hat{\,}(\bdp)}(\Re g)\hat{\,}(\bdp) + \left(\o_m(\bdp)-|\bdp|\right)\overline{(\Im f)\hat{\,}(\bdp)}(\Im g)\hat{\,}(\bdp)\right].
\end{multline}
Thus we see that if we define $\cH_m^\pm(I) := \overline{\cD_0(I,\bR)}^{\|\cdot\|_{m,\pm 1}}$, there holds $\overline{\cD_0(I)}^{\|\cdot\|_m} \cong \cH_m^-(I) \oplus \cH_m^+(I)$, and in this decomposition
\[
\Id - T = \left(\begin{matrix} Q^- &0 \\ 0 &Q^+ \end{matrix}\right)
\]
with bounded $Q^\pm : \cH_m^\pm(I) \to \cH_m^\pm(I)$.

\begin{Lemma}
There holds
\[
\langle f,Q^-g\rangle_{m,-1} = \int_{I \times I} d\bdx d\bdy\,f(\bdx)g(\bdy) Q^-(\bdx-\bdy), \qquad f,g \in \cD_0(I,\bR),
\]
with
\[
Q^-(\bdx) := K_0(m|\bdx|) +  \log |\bdx|, \qquad \bdx \in \bR,
\]
where $K_0$ is the modified Bessel function of order zero. Moreover $Q^-$ is continuously differentiable and with second derivative in $L^2_{loc}(\bR)$. 
\end{Lemma}

\begin{proof}
From~\eqref{eq:1-T} we see that $\langle f,Q^-g\rangle_{m,-1} $ is the difference between the scalar field two-point functions for masses $m>0$ and zero evaluated on the time zero line, and from, e.g.,~\cite[9.6.21]{Ol}, we obtain
\[\begin{split}
\langle f,Q^-g\rangle_{m,-1}  &=  \frac 1 2\int_\bR d\bdp\,\left(\frac 1 {\o_m(\bdp)}-\frac 1 {|\bdp|}\right)\overline{\hat f(\bdp)}\hat g(\bdp) \\
&= \int_{I \times I} d\bdx d\bdy\,f(\bdx) g(\bdy)K_0(m|\bdx-\bdy|) - \frac 1 2\int_\bR d\bdp\,\frac 1 {|\bdp|}\overline{\hat f(\bdp)}\hat g(\bdp).
\end{split}\]
In order to treat the second term, we define functions $F, G \in \cD(I,\bR)$ by $F(\bdx) := \int_{\infty}^\bdx d\bdy \,f(\bdy)$, $G(\bdx) := \int_{-\infty}^\bdx d\bdy\,g(\bdy)$, so that $\hat f(\bdp) =   -i\bdp \hat F(\bdp)$, $\hat g(\bdp) = -i\bdp\hat G(\bdp)$. With this definition, we get
\[\begin{split}
\int_\bR d\bdp\,\frac 1 {|\bdp|}\overline{\hat f(\bdp)}\hat g(\bdp) &= \int d\bdp\,|\bdp|\overline{\hat F(\bdp)}\hat G(\bdp) \\
&= \lim_{\varepsilon \to 0^+} \int_{-\infty}^0 d\bdp |\bdp| e^{\varepsilon \bdp}\overline{\hat F(\bdp)}\hat G(\bdp) +\int_{0}^{+\infty} d\bdp |\bdp| e^{-\varepsilon \bdp}\overline{\hat F(\bdp)}\hat G(\bdp) \\
&= \lim_{\varepsilon \to 0^+} \int_{I \times I} d\bdx d\bdy F(\bdx) G(\bdy)\left[-\int_{-\infty}^0 d\bdp \,\bdp e^{i\bdp(\bdy-\bdx-i\varepsilon)}+\int_{0}^{+\infty} d\bdp\, \bdp e^{i\bdp(\bdy-\bdx+i\varepsilon)}\right]\\
&=- \lim_{\varepsilon \to 0^+} \int_{I \times I} d\bdx d\bdy F(\bdx) G(\bdy)\left[\frac{1}{(\bdy-\bdx-i\varepsilon)^2}+\frac{1}{(\bdy-\bdx+i\varepsilon)^2}\right]\\
&= -\lim_{\varepsilon \to 0^+} \int_{I \times I} d\bdx d\bdy F(\bdx) G(\bdy)\frac{\partial^2}{\partial \bdx\partial \bdy}\left[\log[(\bdx-\bdy)^2+\varepsilon^2]\right] \\
&= -\lim_{\varepsilon \to 0^+} \int_{I \times I} d\bdx d\bdy f(\bdx) g(\bdy)\log[(\bdx-\bdy)^2+\varepsilon^2]\\
&= -2\int_{I \times I} d\bdx d\bdy f(\bdx) g(\bdy)\log|\bdx-\bdy|,
\end{split}\]
where the last equality is obtained by the dominated convergence theorem, observing that for $(\bdx - \bdy)^2 +\varepsilon^2 < 1$ the inequality
\[
|f(\bdx) g(\bdy)\log[(\bdx-\bdy)^2+\varepsilon^2]| \leq |f(\bdx) g(\bdy)\log[(\bdx-\bdy)^2]|
\]
holds, and the right hand side is an integrable function. 

Finally, the last statement is obtained from the fact that there exist analytic functions $\varphi_1$, $\varphi_2$ defined in a neighbourhood of the origin, such that~\cite[9.6.12-13]{Ol}
\[
Q^-(\bdx) = -m^2 \bdx^2\log(m|\bdx|)\varphi_1(m^2\bdx^2)+\varphi_2(m^2\bdx^2), 
\]
and moreover $\bdx \mapsto \log(|\bdx|)$ belongs to $L^2_{loc}(\bR)$.
\end{proof}

\begin{Lemma}\label{Q-}
The operator $Q^-$ is of trace class on $\cH^-_m(I)$.
\end{Lemma}

\begin{proof}
We can of course assume $I =(-1,1)$. By the change of variables $\bdxi := \bdx -\bdy$, $\bdet := \bdx+\bdy$ we can write
\[
\langle f, Q^-g\rangle_{m,-1} = \frac 1 2 \int_{-2}^2 d\bdxi \left[\int_{-2+|\bdxi|}^{2-|\bdxi|} d\bdet\, f\Big(\frac{\bdet+\bdxi}{2}\Big) g\Big(\frac{\bdet-\bdxi}{2}\Big)\right] Q^-(\bdxi)\varphi(\bdxi),
\]
with $\varphi \in \cD(\bR,\bR)$ such that $\varphi(\bdxi) = 1$ for $|\bdxi| \leq 2$ and $\varphi(\bdxi) = 0$ for $|\bdxi| \geq 3$. We now expand the function $Q^- \varphi$ in Fourier series on the interval $[-4.4]$, obtaining
\[
Q^-(\bdxi)\varphi(\bdxi) = \sum_{k \in \bZ} Q_k^- e^{i\frac{\pi k}{4}\bdxi},
\]
with a totally convergent expansion, that is $\sum_{k \in \bZ} |Q_k^-| < \infty$, since the periodized $Q^- \varphi |_{[-4,4]}$ is in $C^1(\bR)$. Actually, $Q_k^- = -\frac{{Q_k}''}{k^2}$, with $Q''_k$ the Fourier coefficients of $\frac{d^2}{d\bdxi^2}(Q^-\varphi) \in L^2([-4,4])$ (by the previous Lemma), and therefore, by the Cauchy-Schwartz and Bessel inequalities,
\[
\sum_{k \in \bZ\setminus\{0\}} |Q_k^-| |k| = \sum_{k \in \bZ\setminus\{0\}} \frac{|Q''_k|}{|k|} \leq \bigg[\sum_{k \in \bZ\setminus\{0\}} |Q''_k|^2\bigg]^{\frac 1 2} \bigg[\sum_{k \in \bZ\setminus\{0\}} \frac 1 {k^2}\bigg]^{\frac 1 2} < +\infty. 
\]
It is therefore legitimate to interchange the series with the $\bdxi$-integration:
\[
\langle f, Q^-g\rangle_{m,-1} = \sum_{k \in \bZ} Q_k^- \int_{I} d\bdx\, f(\bdx) e^{i\frac{\pi k}{4}\bdx} \int_I d\bdy\, g(\bdy) e^{-i\frac{\pi k}{4}\bdy}.
\]
Let now $\chi \in \cD(\bR,\bR)$ be such that $\chi(\bdx) = 1$ for $|\bdx| \leq 1$ and $\chi(\bdx) = 0$ for $|\bdx| \geq 2$. Then
\[\begin{split}
\int_{I} d\bdx\, f(\bdx) e^{i\frac{\pi k}{4}\bdx} &= \int_{I} d\bdx\, f(\bdx) \chi(\bdx)\Big[\cos\Big(\frac{\pi k}{4}\bdx\Big)+ i \sin\Big(\frac{\pi k}{4}\bdx\Big)\Big]\\
&=\frac 1 {2\pi}\left[\int_{\bR} \frac{d\bdp}{\o_m(\bdp)} \overline{\hat f(\bdp)} \hat \psi_k^-(\bdp)+i\int_{\bR} \frac{d\bdp}{\o_m(\bdp)} \overline{\hat f(\bdp)} \hat \varphi_k^-(\bdp)\right],
\end{split}\]
with $\hat \psi_k^-(\bdp) := \frac{\o_m(\bdp)}{2}[ \hat \chi(\bdp+\frac{\pi k}{4})+\hat \chi(\bdp-\frac{\pi k}{4})]$, $\hat \varphi_k^-(\bdp) := \frac{\o_m(\bdp)}{2i}[ \hat \chi(\bdp+\frac{\pi k}{4})-\hat \chi(\bdp-\frac{\pi k}{4})]$. Extending the definitions of $\|\cdot\|_{m,-1}$ and $\langle \cdot, \cdot \rangle_{m,-1}$ to $\cS(\bR,\bR)$, and observing that, for all $k \in \bZ$,
\[\begin{split}
\int_\bR d\bdp\, \o_m(\bdp) \Big|\hat \chi\Big(\bdp+\frac{\pi k}{4}\Big)\Big|^2 &= \int_\bR d\bdp\, \o_m\Big(\bdp-\frac{\pi k}{4}\Big) |\hat \chi (\bdp)|^2\\
&\leq C \int_\bR d\bdp\, \Big(1+\Big|\bdp-\frac{\pi k}{4}\Big|\Big) |\hat \chi (\bdp)|^2 \leq C_1 + C_2 |k|,
\end{split}\]
for suitable constants $C_1, C_2  > 0$, one gets the estimates $\|\psi_k^-\|_{m,-1}, \|\varphi_k^-\|_{m,-1} \leq \sqrt{C_1+C_2|k|}$. Thus in particular $\psi_k^-,\varphi_k^- \in \overline{\cS(\bR,\bR)}^{\|\cdot\|_{m,-1}}$ for all $k \in \bZ$, and therefore, if $P^- : \overline{\cS(\bR,\bR)}^{\|\cdot\|_{m,-1}} \to \cH_m^-(I)$ denotes the orthogonal projection,
\[
\int_{I} d\bdx\, f(\bdx) e^{i\frac{\pi k}{4}\bdx} = \frac 1 {2\pi}[\langle f, \psi_k^-\rangle_{m,-1}+i\langle f,\varphi_k^-\rangle_{m,-1}] =\frac 1 {2\pi} [\langle f, P^-\psi_k^-\rangle_{m,-1}+i\langle f, P^-\varphi_k^-\rangle_{m,-1}],
\]
and, taking into account the fact that $\langle f, Q^- g\rangle_{m,-1}$ is real,
\[\begin{split}
\langle f, Q^- g\rangle_{m,-1} = \frac 1{4\pi^2}  \sum_{k \in \bZ} &\big[\Re Q_k^-\big( \langle f,P^-\psi_k^-\rangle_{m,-1} \langle P^-\psi_k^-,g\rangle_{m,-1}+\langle f,P^-\varphi_k^-\rangle_{m,-1} \langle P^-\varphi_k^-,g\rangle_{m,-1}\big)\\
&-\Im Q_k^-\big( \langle f,P^-\varphi_k^-\rangle_{m,-1} \langle P^-\psi_k^-,g\rangle_{m,-1}+\langle f,P^-\psi_k^-\rangle_{m,-1} \langle P^-\varphi_k^-,g\rangle_{m,-1}\big)\big].
\end{split}\]
Now, since by the above estimates
\[
\sum_{k \in \bZ} |Q_k^-| \|P^-\psi_k^-\|_{m,-1}^2 \leq  C_1 \sum_{k \in \bZ} |Q_k^-|+C_2\sum_{k \in \bZ} |Q_k^-| |k| < +\infty,
\]
and similarly for $\sum_{k \in \bZ} |Q_k^-| \|P^-\varphi_k^-\|_{m,-1}^2$, $\sum_{k \in \bZ} |Q_k^-| \|P^-\psi_k^-\|_{m,-1}\|P^-\varphi_k^-\|_{m,-1}$, and recalling that $\cD_0(I,\bR)$ is dense in $\cH_m^-(I)$, we obtain the thesis using~\cite[Thm.\ 7.12]{We}.
\end{proof}

\begin{Lemma}
There holds
\[
\langle f,Q^+g\rangle_{m,1} = \int_{I \times I} d\bdx d\bdy\,f(\bdx)g(\bdy) Q^+(\bdx-\bdy), \qquad f,g \in \cD_0(I,\bR),
\]
with
\[
Q^+(\bdx) :=  -\frac{m}{|\bdx|}K_1(m|\bdx|) +  \frac{1}{\, |\bdx|^2}, \qquad \bdx \in \bR,
\]
where $K_1$ is the modified Bessel function of order one. Moreover $Q^+$ is in $L^2_{loc}(\bR)$. 
\end{Lemma}

\begin{proof}
We start observing that, by differentiating under the integral sign, one gets
$$
\frac{\partial}{\partial m} \int_\bR d\bdp\, {\o_m(\bdp)} \overline{\hat f(\bdp)}\hat g(\bdp) = m \int_\bR \,\frac {d\bdp} {\o_m(\bdp)}\overline{\hat f(\bdp)}\hat g(\bdp)  \ , \quad m > 0 \ . 
$$
Since the r.h.s. is integrable in a right neighbourhood of $m=0$, 
\begin{align*}
\langle f,Q^+g\rangle_{m,1} & = \frac 1 2 \int_\bR d\bdp\, ({\o_m(\bdp)}-|\bdp|) \overline{\hat f(\bdp)}\hat g(\bdp) = \frac 1 2  \int_0^m dm' m' \int_\bR \,\frac {d\bdp} {\o_{m'}(\bdp)}\overline{\hat f(\bdp)}\hat g(\bdp) \\
& =  \int_0^m dm' m' \int_{I \times I} d\bdx d\bdy\,f(\bdx) g(\bdy)K_0(m'|\bdx-\bdy|) \ . 
\end{align*}
By \cite[9.6.13]{Ol}, we have $|K_0(z)| \leq C_1 |\log z| + C_2$, $|z| \leq 2m$, for suitable constants $C_1$ and $C_2$ and, since 
$$\int_0^m dm' m' |\log(m'|\bdx-\bdy|)|= \begin{cases} \frac 1 {2|\bdx - \bdy|^2} + \frac{m^2} 2 \left( \log(m|\bdx-\bdy|) - \frac 1 2\right) & |\bdx - \bdy| \geq \frac 1 m \\ 
- \frac{m^2} 2 \left( \log(m|\bdx-\bdy|) - \frac 1 2\right) & |\bdx - \bdy| \leq \frac 1 m \end{cases}$$
we can apply Fubini's Theorem and interchange the order of the integrals, thus obtaining
$$\langle f,Q^+g\rangle_{m,1}  = \int_{I \times I} d\bdx d\bdy\,f(\bdx) g(\bdy)  \int_0^m dm' m'  K_0(m'|\bdx-\bdy|) \ . $$
Making use of \cite[9.6.28]{Ol}, we see that $m' K_0(m'|\bdx-\bdy|) = \frac 1 {|\bdx - \bdy|^2} \frac \partial {\partial m'} \big(-m' |\bdx - \bdy| K_1(m' |\bdx - \bdy|) \big)$
and therefore, also observing that $\lim_{z \to 0} z K_1(z) = 1$ by \cite[9.6.11]{Ol}, we find the required expression for $Q^+$ as in the statement.
 
 Finally, the last claim readily follows again from \cite[9.6.11]{Ol} according to which 
 $$Q^+(\bdx) =  \varphi_1(m^2 \bdx^2) + \log(m |\bdx|)\varphi_2(m^2 \bdx^2)$$
for suitable analytic functions $\varphi_1$ and $\varphi_2$.
\end{proof}

\begin{Lemma}\label{Q+}
The operator $Q^+$ is of trace class on $\cH^+_m(I)$.
\end{Lemma}

\begin{proof}
As in the proof of Lemma \ref{Q-}, we assume that $I=(-1,1)$. We also employ $\varphi$ and $\chi$ as defined there. We define 
$F(\bdxi) = \int_{-2+|\bdxi|}^{2-|\bdxi|} d\bdet\, f\Big(\frac{\bdet+\bdxi}{2}\Big) g\Big(\frac{\bdet-\bdxi}{2}\Big)$ and $R(\bdxi) = \int_{-2}^\bdxi d\bdxi' Q^+(\bdxi').$
Therefore,
\begin{equation}\label{eq:Q+}
\langle f, Q^+g\rangle_{m,1} = \frac 1 2 \int_{-2}^2 d\bdxi F(\bdxi) Q^+(\bdxi) = - \frac 1 2 \int_{-2}^2 d\bdxi F'(\bdxi)R(\bdxi) \varphi(\bdxi)
\end{equation}
where we used integration by parts, the fact that $F(\pm 2)=0$ and that $F'(\bdxi)$ exists for all $\bdxi \neq 0$ and is integrable in $[-2,2]$.
We expand the function $R \varphi$ in Fourier series on the interval $[-4.4]$, obtaining
\[
R(\bdxi)\varphi(\bdxi) = \sum_{k \in \bZ} R_k e^{i\frac{\pi k}{4}\bdxi},
\]
and notice that, again integrating by parts, $R_k = \frac{1}{2\pi i k} \int_{-4}^4 d\bdxi (R\varphi)' (\bdxi) e^{-i\frac{\pi k}{4}\bdxi} =: \frac{ 4 Q^+_k}{i \pi k}$, $k \neq 0$.
Since $(R\varphi)' \in L^2([-4,4])$ by the previous lemma, it follows that $\sum_k |Q^+_k|^2 < +\infty$
 and hence
 $$\sum_{k \neq 0} |R_k| = \frac 4 \pi \sum_{k \neq 0}  \frac{|Q^+_k|}{|k|} \leq \frac 4 \pi \left(\sum_{k \neq 0}  |Q^+_k|^2\right)^{1/2} \left( \sum_{k \neq 0}  \frac 1 {k^2} \right)^{1/2} < +\infty \  . $$
 We are then allowed to interchange the integral with the series in (\ref{eq:Q+}) and obtain
 \begin{align*}
\langle f, Q^+g\rangle_{m,1} & = - \frac 1 2 \sum_k R_k \int_{-2}^2 d\bdxi F'(\bdxi) e^{i\frac{\pi k}{4}\bdxi} =  \frac 1 2 \sum_k R_k \frac{i \pi k} 4 \int_{-2}^2 d\bdxi F(\bdxi) e^{i\frac{\pi k}{4}\bdxi} \\
 & = \sum_{k\neq 0} Q^+_k \int_{I} d\bdx\, f(\bdx) e^{i\frac{\pi k}{4}\bdx} \int_I d\bdy\, g(\bdy) e^{-i\frac{\pi k}{4}\bdy}. 
 \end{align*}
 Now, similarly to the proof of Lemma \ref{Q-},
 $$
\int_{I} d\bdx\, f(\bdx) e^{i\frac{\pi k}{4}\bdx}
=\frac 1 {2\pi}\left[\int_{\bR} d\bdp \, \o_m(\bdp) \overline{\hat f(\bdp)} \hat \psi_k^+(\bdp)+i\int_{\bR} d\bdp \, \o_m(\bdp) \overline{\hat f(\bdp)} \hat \varphi_k^+(\bdp)\right],
$$
with $\hat \psi_k^+(\bdp) := \frac 1 {2 \o_m(\bdp)}[ \hat \chi(\bdp+\frac{\pi k}{4})+\hat \chi(\bdp-\frac{\pi k}{4})]$, $\hat \varphi_k^+(\bdp) := \frac 1 {2i \o_m(\bdp)}[ \hat \chi(\bdp+\frac{\pi k}{4})-\hat \chi(\bdp-\frac{\pi k}{4})]$.
 Observe that, for all $\bdp \in \bR$,
 $$\sup_{\bdq \in \bR} \frac{|\bdq|}{\o_m(\bdp-\bdq)} = \sup_{\bdq \in \bR} \frac{|\bdp - \bdq|}{\o_m(\bdq)} \leq C(1+|\bdp|)$$
for a suitable $C>0$.
Therefore, for all $k \in \bZ \setminus \{0\}$,
$$\int_\bR d\bdp \frac 1 {\o_m(\bdp)} \Big|\hat \chi\Big(\bdp+\frac{\pi k}{4}\Big)\Big|^2 = \int_\bR \frac{d\bdp} {\o_m\Big(\bdp-\frac{\pi k}{4}\Big)} |\hat \chi (\bdp)|^2 \leq \frac{4 C}{\pi |k| } \int_\bR d\bdp (1 + |\bdp|)  |\hat \chi (\bdp)|^2  \ , $$
which implies that $\|\psi_k^+\|_{m,1}, \|\varphi_k^+\|_{m,1} \leq \frac{C'}{\sqrt{ |k|}}$, $k \neq 0$.
Thus, we obtain that $\sum_k |Q_k^+| \, \|P^+  \psi_k^+\|_{m,1}^2$, $\sum_k |Q_k^+| \, \|P^+  \varphi_k^+\|_{m,1}^2$ and $\sum_k |Q_k^+| \, \|P^+  \psi_k^+\|_{m,1}  \|P^+  \varphi_k^+\|_{m,1}$ are all convergent,
and the conclusion follows as in the proof of Lemma \ref{Q-}.
\end{proof}

\textbf{Acknowledgements.} We are grateful to D.~Buchholz for useful discussions about the subject of this work and for comments on a preliminary version, and to E.~Valdinoci for pointing out reference~\cite{Ma}. R.~C.\ is partially supported by the Sapienza Ricerca Scientifica 2017 grant ``Algebre di Operatori e Analisi Armonica Noncommutativa''. G.~M.\ is partially supported by the MIUR Excellence Department Project awarded to the Department of Mathematics, University of Rome Tor Vergata, CUP E83C18000100006, the ERC Advanced Grant 669240 ``Quantum Algebraic Structures and Models'', the INDAM-GNAMPA, and the Tor Vergata University grant ``Operator Algebras and Applications to Noncommutative Structures in Mathematics and Physics''.


\end{document}